\documentclass{amsart}[11pt]
\usepackage{amsmath}
\usepackage{amssymb}
\usepackage{amsthm}
\usepackage{mathrsfs}
\usepackage[margin=1in]{geometry}
\usepackage{color}
\usepackage{graphicx}
\usepackage{amsfonts}
\usepackage{hyperref}
\usepackage{enumerate}
\usepackage{xcolor}

\newtheorem{theorem}{Theorem}[section]
\newtheorem{corollary}[theorem]{Corollary}
\newtheorem{lemma}[theorem]{Lemma}
\newtheorem{proposition}[theorem]{Proposition}

\theoremstyle{definition}
\newtheorem{definition}[theorem]{Definition}
\newtheorem{remark}[theorem]{Remark}
\newtheorem{assumption}[theorem]{Assumption}

\newtheorem{example}[theorem]{Example}

\newcommand{\ind}{1\hspace{-2.1mm}{1}}

\newcommand{\RR}{\mathbb{R}}
\newcommand{\PP}{\mathbb{P}}
\newcommand{\EE}{\mathbb{E}}

\newcommand{\Vv}{\mathcal{V}}
\newcommand{\Cc}{\mathcal{C}}

\newcommand{\XX}{\mathbb{X}}
\newcommand{\Ll}{\mathcal{L}}
\newcommand{\Oo}{\mathcal{O}}
\newcommand{\Yy}{\mathcal{Y}}
\newcommand{\Zz}{\mathcal{Z}}

\newcommand{\Di}{\mathrm{D}}
\newcommand{\D}{\mathrm{d}}

\newcommand{\E}{\mathrm{e}}

\newcommand{\QQ}{\mathbb{Q}}

\newcommand{\loc}{\mathrm{L}}

\newcommand{\ar}{\mathrm{a}}
\newcommand{\bb}{\mathrm{b}}
\newcommand{\xx}{\mathrm{x}}

\newcommand{\Ff}{\mathcal{F}}

\newcommand{\rrho}{\overline{\rho}}
\newcommand{\eps}{\varepsilon}
\newcommand{\be}{\begin{equation}}
\newcommand{\ee}{\end{equation}}

\newcommand{\T}{\mathsf{T}}

\begin{document}

\title{Pathwise moderate deviations for option pricing}
\author{Antoine Jacquier}
\address{Department of Mathematics, Imperial College London and Baruch College, CUNY}
\email{a.jacquier@imperial.ac.uk}
\author{Konstantinos Spiliopoulos}
\address{Department of Mathematics and Statistics, Boston University}
\email{kspiliop@math.bu.edu}
\thanks{K.S. was partially supported by NSF DMS 1550918.}

\date{\today}
\maketitle

\begin{abstract}
We provide a unifying treatment of pathwise moderate deviations for models commonly used in financial applications,
and for related integrated functionals.
Suitable scaling allows us to transfer these results into small-time, large-time and tail asymptotics
for diffusions, as well as for option prices and realised variances.
In passing, we highlight some  intuitive relationships between moderate deviations rate functions and their large deviations counterparts;
these turn out to be useful for numerical purposes, as large deviations rate functions are often difficult to compute.
\end{abstract}
\section{Introduction}

We develop a unifying framework for pathwise moderate deviations of (multiscale) It\^o diffusions,
with applications to small-time, large-time and tail asymptotics for the diffusions and related integrated functionals.
The original motivation is a pathwise extension of the moderate deviations proved in~\cite{FGP}
in the context of small-time option pricing only.
More specifically, we consider a generic two-dimensional stochastic volatility model $\XX = (X,Y)$ of the form
\begin{equation}\label{eq:MainIntro}
\begin{array}{rl}
\D X_t & = \displaystyle -\frac{1}{2}\sigma^2(Y_t) \D t + \sigma(Y_t)\D W_t,\\
\D Y_t & = f(Y_t) \D t + g(Y_t)\D Z_t,\\
\D\langle W, Z\rangle_t & = \rho \D t,
\end{array}
\end{equation}
with starting point $(X_0, Y_0) = (x_0, y_0)$.
For $0<\eps,\delta\leq 1$, the rescaled process~$\XX^\eps$ defined by
\begin{equation}\label{Eq:TimeSpaceTransformationIntro}
\XX^\eps_t := (X^\eps_t, Y^\eps_t) := \left(\delta X_{\frac{\eps}{\delta^{2}}t}, Y_{\frac{\eps}{\delta^{2}}t}\right),
\end{equation}
is a solution of the multiscale system
\begin{equation}\label{eq:MainIntroEps}
\begin{array}{rl}
\D X^{\eps}_t & = \displaystyle -\frac{1}{2}\frac{\eps}{\delta}\sigma^2(Y^{\eps}_t) \D t + \sqrt{\eps}\sigma(Y^{\eps}_t)\D W_t,\\
\D Y^{\eps}_t & = \displaystyle \frac{\eps}{\delta^{2}}f(Y^{\eps}_t) \D t + \frac{\sqrt{\eps}}{\delta}g(Y^{\eps}_t)\D Z_t,
\end{array}
\end{equation}
starting from $\XX^\eps_0 = (\delta x_0, y_0)$, again with $\D\langle W, Z\rangle_t = \rho \D t$.
One reasonably expects that if $\delta=1$, then small~$\eps$ asymptotics for~\eqref{eq:MainIntroEps}
correspond to small-time asymptotics for~\eqref{eq:MainIntro}.
Equivalently, if we allow both~$\eps$ and~$\delta$ to tend to zero such that $\lim\frac{\eps}{\delta}=\gamma\in(0,\infty)$,
then small $(\eps,\delta)$ asymptotics for~\eqref{eq:MainIntroEps} correspond to scaled large-time asymptotics for~\eqref{eq:MainIntro}.
Hence, one can describe small- and large-time asymptotics for~\eqref{eq:MainIntro} via small $\eps$ asymptotics for~\eqref{eq:MainIntroEps} by appropriately choosing $\delta$.
Such asymptotic results can then by translated to corresponding estimates for option pricing,
a quantity of interest in mathematical finance.
In addition, one more quantity of interest in financial applications is related to asymptotics of integrated functionals of the form $\int_{0}^{t}H(X^{\eps}_{s},Y^{\eps}_{s})\D s$ for appropriate functions~$H$.
In particular, $H(x,y)=y$ corresponds to the so-called realised variance.
The focus of this paper is to quantify such approximations via the lens of moderate deviations (MD),
and we shall analyse three situations.

First, considering $\overline{\XX}_{t}:=\lim_{\eps\downarrow 0} \XX^{\eps}_{t}$ (in some appropriate sense)
and a positive function $h(\cdot)$ such that $\lim_{\eps\downarrow 0}h(\eps)=\infty$ with $\lim_{\eps\downarrow 0}\sqrt{\eps} h(\eps)=0$, we are interested in the large deviations properties of
$$
\eta^{\eps}_{t} := \frac{\XX^{\eps}_{t}-\overline{\XX}_{t}}{\sqrt{\eps}h(\eps)},
$$
or equivalently moderate deviations properties for $\XX^{\eps}$,
and, as a corollary, for~$X^\eps$,
and their consequences on estimates for option prices on~$\E^{X}$.
Moderate deviations for~$X^{\eps}$ have been developed in the literature both in the $\delta=1$
case~\cite{BaierFreidlin, Freidlin1978}, and in the $\eps,\delta\downarrow 0$ case~\cite{BaierFreidlin, Freidlin1978,G03,MS};
we make here precise the connection of those results to small- and large-time moderate deviations results respectively for~\eqref{eq:MainIntro},
and provide exact consequences for option pricing.
The mathematical finance literature has seen an explosion of applications of large deviations
over the past decade~\cite{DFJV1, DFJV2, FFK12, FordeJac, FGGJT, JKRM, Robertson},
but moderate deviations have only recently been introduced by~\cite{FGP}.
However, \cite{FGP} only considers small-time behaviours, which are not pathwise
and are based on assuming that~$X$ admits a continuous density expansion at all times with a specific behaviour for small times. Our methodology here allows us to derive pathwise results without the need for such assumptions.

Second, we obtain moderate deviations principle for some integrated functionals,
more precisely large deviations for quantities of the form
$$
R^{\eps} := \frac{1}{\sqrt{\eps}h(\eps)}\int_{0}^{\cdot}\left(H(X^{\eps}_{s},Y^{\eps}_{s})-\overline{H}(X^{\eps}_{s})\right)\D s,
$$
for some function~$H$, where~$\overline{H}$ centers~$H$
around its mean behaviour with respect to the invariant distribution of the process~$Y$.
We also allow the coefficients of the SDE~\eqref{eq:MainIntro} to depend on both $(x,y)$ and allow certain growth with respect to $y$ (namely we do not impose global boundedness assumptions). This constitutes the main theoretical result of this paper, presented in Section~\ref{S:MDPIntegratedFunctionals}
and proven in Sections~\ref{S:ControlRep}-\ref{S:ProofIntegratedFunctMDP}.
For example, when $H(x,y)=y$, this corresponds to asymptotics of the realised variance and yields estimates for options thereon.
Related moderate deviations results have also been established in~\cite{GuillinLipster2005},
but with~$H$ uniformly bounded, which is insufficient for many applications.
Our proof is based on the weak convergence approach of~\cite{DE97} through stochastic control representations,
which then allow us to consider growth in the coefficients of both the underlying SDE models and on~$H$.

Finally, we show how to apply the pathwise moderate deviations to models used in quantitative finance,
and derive asymptotics for option prices.
In passing, we obtain interesting connections between the large deviations and the moderate deviations rate functions.
The latter in fact characterises the local curvature of the large deviations rate function around its minimum.
Even though this is to be intuitively expected, it shows that one can use moderate deviations to get an approximation to large deviations asymptotics,
and is useful in practice as the moderate deviations rate function is usually available in closed form as a quadratic,
whereas the large deviations rate function is often not explicit.
Lastly, we mention that even though the results in this paper are mostly presented for $\XX\in\RR^2$,
the actual theorems hold in any finite dimension, as in~\cite{BaierFreidlin,ChiariniFischer2014, Freidlin1978,MS}.
We restrict ourselves here to this setting, solely because of the financial motivation.

The rest of the paper is organised as follows.
In Section~\ref{S:SmallTimeMDP} we review the small-noise pathwise moderate deviation results from the literature,
and connect them to small-time moderate deviations asymptotics for financial models.
In Section~\ref{S:LargeTimeMDP} we review the small-noise moderate deviations
for slow-fast systems~\eqref{eq:MainIntroEps}, with $\lim\frac{\eps}{\delta}=\gamma\in(0,\infty)$,
and connect them to the corresponding large deviations estimates. 
In Section~\ref{S:MDPIntegratedFunctionals} we present our main new theoretical result, on moderate deviations for integrated functionals.
Section~\ref{S:FinancialApplications} contains the financial implications of our results on option price asymptotics,
Sections~\ref{S:ControlRep}-\ref{S:ProofIntegratedFunctMDP} gather the proof of the moderate deviations for integrated functionals,
while Appendices~\ref{S:Regularity} and~\ref{S:PrelResControlledProcesses}
state some results of independent interest, used throughout the paper.

\section{Pathwise moderate deviations for small-noise diffusions}\label{S:SmallTimeMDP}
We first consider pathwise moderate deviations for small-noise It\^o diffusions, following~\cite{BaierFreidlin, Freidlin1978,MS}.
By suitable rescaling, this allows us to unify and generalise several recent results
in the mathematical finance literature on small-time moderate deviations, in particular those of~\cite{FGP}. However, before proceeding let us recall the definition of a large deviations principle.
\begin{definition}\label{Def:LDP}
Let $\mathcal{P}$ be a Polish space and $\mathbb{P}$ be a probability measure on $(\mathcal{P},\mathcal{B}(\mathcal{P}))$, we say that a collection $(Z^{\eps})_{\eps\in (0,1)}$
of $\mathcal{P}$-valued random variables has a large deviations principle with rate function $S:\mathcal{P}\to [0,\infty]$ and speed $1/\eps$ if
\begin{enumerate}
\item For each $M\geq 0$, the set $\Phi(M) = \left\{ \xi\in \mathcal{P}: S(\xi)\leq M\right\}$
is a compact subset of $\mathcal{P}$.
\item For every open $G\subset \mathcal{P}$,
\begin{equation*}
\liminf_{\eps\downarrow 0}\eps\ln \mathbb{P}\left(Z^{\eps}\in G\right) \geq -\inf_{\xi\in G} S(\xi)
\end{equation*}
\item For every closed $F\subset \mathcal{P}$,
\begin{equation*}
\limsup_{\eps\downarrow 0}\eps\ln \mathbb{P}\left( Z^{\eps}\in F\right)\leq -\inf_{x\in F} S(\xi).
\end{equation*}
\end{enumerate}
\end{definition}

 A fact that we will heavily use in this paper is that Definition~\ref{Def:LDP} implies that for a regular set $\Gamma\subset \mathcal{B}(\mathcal{P})$, i.e.~$\Gamma$ is such that the infimum of~$S$
 over the closure~$\text{cl}(\Gamma)$ is the same as the infimum of~$S$ over the interior~$\Gamma^{o}$,  we have that
\[
\lim_{\eps\downarrow 0}\mathbb{P}\left(Z^{\eps}\in \Gamma\right)=-\inf_{\xi\in\text{cl}(\Gamma)}S(\xi).
\]

Consider now the $\RR^d$-valued SDE
\begin{equation}\label{Eq:GeneralSDEModel}
\D \XX_{t}^{\eps} = \bb_{\eps}(t,\XX_{t}^{\eps})\D t+\sqrt{\eps}\ar_{\eps}(t,\XX_{t}^{\eps})\D W_{t},
\qquad \XX_{0}^{\eps} = \xx_{0}\in\RR^d,
\end{equation}
where for every $\eps\in(0,1)$, $\bb_{\eps}:\RR_{+}\times\RR^{d}\to\RR^{d}$,
$\ar_{\eps}:\RR_{+}\times\RR^{d}\to \mathcal{M}^{d\times m}$
(the space of real-valued $d\times m$ matrices)
 and~$W$ is a standard multidimensional Brownian motion taking values in $\RR^{m}$ defined on an appropriately filtered probability space.
 We shall also fix some arbitrary time horizon $T>0$.
We will be working with diffusions of the form~\eqref{Eq:GeneralSDEModel} that have unique strong solutions,
and that satisfy the following assumption:
\begin{assumption}\label{A:MainGeneralAssumption}
\
\begin{enumerate}[(i)]
\item The coefficients $\bb_{\eps}$, $\ar_{\eps}$ as well as $\Di\bb_{\eps}$ converge as~$\eps$ tends to zero to some functions $\bb$, $\ar$ and $\Di\bb$ respectively, uniformly on bounded subsets of $\RR_{+}\times\RR^{d}$.
\item
The coefficients $\bb_{\eps},\bb\in \mathcal{C}^{0,1}$, $\ar_{\eps},\ar$ are locally Lipschitz in $x$ uniformly in $\eps\in(0,1)$ and $t\in[0,T]$.
\item There exists $C>0$ such that for all $t\in[0,T]$ and $x\in \RR^{d}$,
\[
\max\left\{\text{Tr}\left[\ar_{\eps} \ar_{\eps}^{\T}(t,x)\right], \langle x,\bb_{\eps}(t,x)\rangle\right\},\max\left\{\text{Tr}\left[\ar \ar^{\T}(t,x)\right], \langle x,\bb(t,x)\rangle\right\}
 \leq C \left(1+|x|^{2}\right),
\]
\end{enumerate}
\end{assumption}

Under Assumption~\ref{A:MainGeneralAssumption}, the SDE~\eqref{Eq:GeneralSDEModel}
has a unique non-explosive solution and $\sup_{t\in[0,T]}|\XX^{\eps}_{t}-\overline{\XX}_{t}|$
converges to zero in probability as~$\eps$ tends to zero,
where~$\overline{\XX}$ satisfies $\overline{\XX}_{t} = x_{0} + \int_{0}^{t}\bb(s,\overline{\XX}_{s})\D s$.

Let now~$h(\cdot)$ be a continuous function on~$(0,1]$, diverging to infinity at the origin,
with $\lim_{\eps\downarrow 0}\sqrt{\eps} h(\eps) = 0$.
Define the moderate deviations process~$\eta^\eps$ pathwise as
\begin{equation} \label{E:deviations}
\eta_t^\eps := \frac{\XX_t^\eps - \overline{\XX}_t}{\sqrt{\eps} h(\eps)},
\qquad\text{for all }t\geq 0.
\end{equation}

We are interested in the moderate deviations principle (MDP) for
$\{\XX^{\eps},\eps>0\}$ in $\Cc\left([0,T];\RR^{d}\right)$,
which amounts to establishing the large deviations principle (LDP)
for $\{\eta^{\eps},\eps>0\}$ in $\Cc\left([0,T];\RR^{d}\right)$.
The following theorem is in the spirit of~\cite{BaierFreidlin,ChiariniFischer2014, Freidlin1978,MS},
and follows by arguments similar to those in~\cite{ChiariniFischer2014, MS}.

\begin{theorem}\label{T:GeneralResult}
Under Assumption~\ref{A:MainGeneralAssumption}, the family $\{\eta^{\eps},\eps>0\}$
satisfies the LDP (or $\{\XX^{\eps},\eps>0\}$ satisfies the MDP) in~$\Cc([0,T];\RR^{d})$
with speed~$h^{2}(\eps)$ and rate function
$$
S(\xi) =
\inf\left\{
\frac{1}{2}\int_{0}^{T}|u_{s}|^{2}\D s:
u\in L^{2}\left([0,T];\RR^{d}\right),
\xi=\int_{0}^{\cdot}\left[\Di\bb(s,\overline{\XX}_{s}) \xi_{s}+\ar(s,\overline{\XX}_{s})u_{s}\right]\D s\right\}.
$$
\end{theorem}

Notice that if $\ar\ar^{\T}(t,x)$ is uniformly non degenerate on the path of $\{\overline{\XX}_s, s\in[0,T]\}$, then we can write
\begin{equation*}
  S(\xi) = \left\{
\begin{array}{ll}
 \displaystyle\frac{1}{2} \int_0^T \left(\dot{\xi}_s - \Di\bb(s,\overline{\XX}_{s}) \xi_{s}  \right)^{\T} \left(\ar\ar^{\T}(s,\overline{\XX}_s)\right)^{-1} \left(\dot{\xi}_s - \Di\bb(s,\overline{\XX}_{s}) \xi_{s}\right) \D s,
  & \text{if }\xi \in \mathcal{AC}([0, T], \RR^d), \xi_{0}=0,\\
 +\infty, & \text{otherwise}.
 \end{array}
 \right.
\end{equation*}
where $\mathcal{AC}([0, T], \RR^d)$ denotes the class of absolutely continuous functions from $[0, T]$ to $\RR^d$.

\subsection{Small-time moderate deviations}\label{SS:SmallTimeMDPFinance}
We now show how to use Theorem~\ref{T:GeneralResult} in order to prove small-time moderate deviations for a class of two-dimensional systems~$\XX$ satisfying the differential equations
\begin{equation}\label{eq:Main}
\begin{array}{rl}
\D X_t & = \displaystyle -\frac{1}{2}\sigma^2(X_t, Y_t) \D t + \sigma(X_t, Y_t)\D W_t,\\
\D Y_t & = f(t, X_t, Y_t) \D t + g(X_t, Y_t)\D Z_t,\\
\D\langle W, Z\rangle_t & = \rho \D t,
\end{array}
\end{equation}
with starting point $\XX_0 = (x_0, y_0)$.
Since we are interested in small-time asymptotics, we rescale~\eqref{eq:Main}
according to~\eqref{Eq:TimeSpaceTransformationIntro} with $\delta=1$,
namely we consider
$\XX^\eps_t := (X_t^{\eps}, Y_t^{\eps}) := (X_{\eps t}, Y_{\eps t})$.
This yields the system
\begin{align*}
\D X_t^\eps & = \displaystyle -\frac{\eps}{2} \sigma^{2}(X_t^\eps, Y_t^\eps) \D t
 + \sqrt{\eps} \sigma(X_t^\eps, Y_t^\eps)\D W_t,\\
\D Y_t^\eps & = \eps f(\eps t,X_t^\eps, Y_t^\eps) \D t + \sqrt{\eps} g(X_t^\eps, Y_t^\eps)\D Z_t,
\end{align*}
with starting point $\XX_0^\eps = (x_0, y_0)$.
The limiting process $\overline{\XX} := \lim_{\eps\downarrow 0}\XX^\eps$ is constant, equal to~$(x_0,y_0)$
 almost surely.
Define now the process $\eta^\eps$ as in~\eqref{E:deviations},
where $\sqrt{\eps}h(\eps)$ tends to zero and $h(\eps)$ tends to infinity as~$\eps$ tends to zero.
In the notations of~\eqref{Eq:GeneralSDEModel}, we have
\begin{equation}\label{Eq:SpecificCoeff}
\bb_{\eps}(t,x,y)=\eps
\begin{pmatrix}
-\frac{1}{2}\sigma^{2}(x,y)\\
 f(\eps t,x,y)
\end{pmatrix}
\qquad\text{and}\qquad
\ar_{\eps}(t,x,y)=
\begin{pmatrix}
\sigma(x,y) & 0\\
 \rho g(x,y) & \rrho g(x,y)
\end{pmatrix},
\end{equation}
where for convenience, we denote $\rrho:=\sqrt{1-\rho^{2}}$.
Then, using Theorem~\ref{T:GeneralResult}, we have the following result:
\begin{proposition}\label{P:MDP_FinancialModelsTwoComp}
If $\bb_{\eps}$ and $\ar_{\eps}$ in~\eqref{Eq:SpecificCoeff} satisfy Assumption~\ref{A:MainGeneralAssumption}, $\sigma(\cdot), g(\cdot)$
are non-zero on~$(x_{0},y_{0})$,
then $(\XX^{\eps})_{\eps>0}$ satisfies the pathwise MDP with speed~$h^{2}(\eps)$ and rate function
\begin{equation*}
S(\phi,\psi) =
\left\{
\begin{array}{ll}
\displaystyle
\frac{1}{2(1-\rho^{2})} \int_0^T \left(\left|\frac{\dot{\phi}_s}{\sigma(x_{0},y_{0})}\right|^{2}
 - \frac{2\rho\dot{\phi}_{s}\dot{\psi}_{s}}{\sigma(x_{0},y_{0}) g(x_{0},y_{0})}
+ \left|\frac{\dot{\psi}_{s}}{g(x_{0},y_{0})}\right|^{2}\right) \D s,
 & \text{if }\phi,\psi \in \mathcal{AC}([0, T],\RR), \\
 & \text{and }(\phi_{0},\psi_{0})=(0,0),\\
 +\infty, & \text{otherwise}.
\end{array}
\right.
\end{equation*}
\end{proposition}

\begin{proof}
The matrix $\ar\ar^{\top}(x_{0},y_{0})$ is invertible, since  by assumption $\sigma(\cdot), g(\cdot)$ are non-zero on~$(x_{0},y_{0})$,
so that Theorem~\ref{T:GeneralResult} applies.
\end{proof}

We are primarily interested in the MDP for the $X^{\eps}$ component,
which we obtain via contraction principle:
\begin{corollary}\label{C:MDP_FinancialModelsOneComp}
Under the assumptions of Proposition~\ref{P:MDP_FinancialModelsTwoComp},
$\{X^{\eps},\eps>0\}$ satisfies the pathwise MDP with speed~$h^{2}(\eps)$ and rate function
\begin{equation*}
I(\phi) =
\left\{\begin{array}{ll}
\displaystyle \frac{1}{2\sigma(x_{0},y_{0})^2} \int_0^T |\dot{\phi}_s|^{2}\D s,
 & \text{if }\phi \in \mathcal{AC}([0, T],\RR) \text{ and } \phi_{0}=0,\\
+\infty, & \text{otherwise}.
\end{array}
\right.
\end{equation*}
\end{corollary}
\begin{proof}
By Proposition~\ref{P:MDP_FinancialModelsTwoComp} and the contraction principle~\cite[Section 4.2.1]{DZ98},
the MDP rate function for the first component of $(X^{\eps}_{t},Y^{\eps}_{t})$ is
\[
I(\phi)=\inf\{S(\phi,\psi): \psi\in \mathcal{C}([0,T],\RR),\psi_{0}=0\},
\]
which can be solved as a variational problem.
For a fixed absolutely continuous trajectory $\phi$, the Euler-Lagrange equation for the stationary path~$\psi$ is
\[
\ddot{\psi}_{t}=\rho\frac{g(x_{0},y_{0})}{\sigma(x_{0},y_{0})}\ddot{\phi}_{t},
\]
or $\dot{\psi}_{t}=\rho\frac{g(x_{0},y_{0})}{\sigma(x_{0},y_{0})}\dot{\phi}_{t}+C$ for some constant~$C$. Plugging it in the expression for~$S(\cdot)$ in Proposition~\ref{P:MDP_FinancialModelsTwoComp},
and minimising over~$C$ we obtain that the minimum is at $C=0$, which concludes the proof.
\end{proof}

\begin{remark}\label{rem:Projection1}
Setting $T=1$, and understanding~$\eps$ as time, Corollary~\ref{C:MDP_FinancialModelsOneComp}
immediately provides us with small-time asymptotics of probabilities.
For $x_1>0$, the minimisation problem trivially yields $\phi_t = x_1 t$ as optimal path on $[0,1]$,
so that
$$
\lim_{\eps\downarrow 0}\frac{1}{h(\eps)^2}\log\PP\left(\frac{X_\eps}{\sqrt{\eps}h(\eps)}\geq x_1\right) = -\frac{x_1^2}{2\sigma(x_{0},y_{0})^2}.
$$

\end{remark}

\subsection{Examples}\label{SS:SmallTimeExamples1}

\subsubsection{Local-stochastic volatility model}
Consider a local stochastic volatility model of the form~\eqref{eq:Main}, with
$$
\sigma(t, X_t, Y_t) = \sigma_{\loc}(t, X_t)\xi(Y_t),
$$
and assume that the local volatility component $\sigma_{\loc}(\eps t, \cdot)$ converges uniformly to~$\overline{\sigma}_{\loc}(\cdot)$ as~$\eps$ tends to zero.
This is consistent with diffusion-type models, but not for jump models~\cite{FGY}.
Applying the time rescaling transformation $t \mapsto \eps t$ yields
\begin{align*}
\D X_t^\eps & = -\frac{\eps}{2}\sigma_{\loc}^2(\eps t, X_t^\eps)\xi(Y_t^\eps)^2 \D t
 + \sqrt{\eps} \sigma_{\loc}(\eps t, X_t^\eps)\xi(Y_t^\eps)\D W_t,\\
\D Y_t^\eps & = \eps f(\eps t, X_t^\eps, Y_t^\eps) \D t + \sqrt{\eps} g(X_t^\eps, Y_t^\eps)\D Z_t,
\end{align*}
with starting point $(X_0^\eps, Y_0^\eps) = (x_0, y_0)$.
Clearly, $\XX^\eps = (X^\eps, Y^\eps)$ converges to the constant process
equal to $\xx_{0} = (x_{0}, y_{0})$
almost surely as~$\eps$ tends to zero.
Choosing $h(\eps) \equiv \eps^{-\beta}$, with $\beta \in (0, 1/2)$, we then obtain that
under Assumption~\ref{A:MainGeneralAssumption}
and the assumptions of Proposition~\ref{P:MDP_FinancialModelsTwoComp},
in particular assuming that $\overline{\sigma}_{\loc}(x_0)\xi(x_0)\ne 0$,
$\eta^\eps := \eps^{\beta-1/2}(\XX^{\eps}-\xx_{0})$ satisfies a LDP as~$\eps$ tends to zero by
Theorem~\ref{T:GeneralResult}, and hence~$X^{\eps}$ satisfies a MDP
from Corollary~\ref{C:MDP_FinancialModelsOneComp} with rate function
\begin{equation*}
I(\phi) =
\left\{
\begin{array}{ll}
\displaystyle \frac{1}{2}\frac{1}{\overline{\sigma}_{\loc}(x_0)^2\xi(y_0)^2}\int_{0}^{T}\dot{\phi}_s^2\D s, & \text{if }\phi \in \mathcal{AC}([0, T],\RR)\text{ and }\phi_0 = 0,\\
+\infty, & \text{otherwise}.
\end{array}
\right.
\end{equation*}

\subsubsection{Stein-Stein}
Consider the Stein-Stein~\cite{Stein} stochastic volatility model:
\begin{equation}\label{eq:SteinStein}
\begin{array}{rll}
\D X_t & = \displaystyle -\frac{1}{2}Y_t^2\D t + Y_t\D W_t, \quad & X_0 = x_{0},\\
\D Y_t & = (a+b Y_t)\D t+c \D Z_t, \quad & Y_0=y_{0},\\
\D\left\langle W,Z\right\rangle_t & = \rho \D t.
\end{array}
\end{equation}
With the time-scaling $t\mapsto \eps t$,
the system~\eqref{eq:SteinStein} reads
\begin{equation*}
\begin{array}{rll}
\D X_t^{\eps} & = \displaystyle -\frac{\eps}{2}(Y_t^{\eps})^2\D t
 + \sqrt{\eps}Y_t^{\eps}\D W_t, \quad & X_0^{\eps}=x_{0},\\
\D Y_t^{\eps} & = (a+b Y_t^{\eps})\eps\D t
 + \sqrt{\eps}c\D Z_t, \quad & Y_0^{\eps} = y_{0},\\
\D\left\langle W,Z\right\rangle_t & = \rho \D t.
\end{array}
\end{equation*}
Clearly, $\XX^\eps = (X^\eps, Y^\eps)$ converges to the constant process
equal to $\xx_{0} = (x_{0}, y_{0})$
almost surely as~$\eps$ tends to zero.
Choosing $h(\eps) \equiv \eps^{-\beta}$, with $\beta \in (0, 1/2)$, we then obtain that
$\eta^\eps := \eps^{\beta-1/2}(\XX^{\eps}-\xx_{0})$ satisfies a LDP as~$\eps$ tends to zero by
Theorem~\ref{T:GeneralResult}, and hence~$X^{\eps}$ satisfies a MDP
from Corollary~\ref{C:MDP_FinancialModelsOneComp} with rate function
\begin{equation*}
I(\phi) =
\left\{
\begin{array}{ll}
\displaystyle \frac{1}{2y_0^2}\int_{0}^{T}\dot{\phi}_s^2\D s, & \text{if }\phi \in \mathcal{AC}([0, T],\RR)
\text{ and }\phi_0 = 0,\\
+\infty, & \text{otherwise}.
\end{array}
\right.
\end{equation*}

\subsubsection{Heston}
Consider the Heston stochastic volatility model:
\begin{equation}\label{eq:Heston}
\begin{array}{rll}
\D X_t & = \displaystyle -\frac{1}{2}Y_t\D t+ \sqrt{Y_t}\D W_t, \quad & X_0=x_{0},\\
\D Y_t & = \kappa(\theta - Y_t)\D t + \xi\sqrt{Y_t} \D Z_t, \quad & Y_0 = y_{0},\\
\D\left\langle W,Z\right\rangle_t & = \rho \D t,
\end{array}
\end{equation}
Applying the time rescaling $t\mapsto \eps t$, the system~\eqref{eq:Heston} turns into
\begin{equation*}
\begin{array}{rll}
\D X_t^{\eps} & = \displaystyle -\frac{\eps}{2}Y_t^{\eps}\D t
 + \sqrt{\eps}\sqrt{Y_t^{\eps}}\D W_t, \quad & X_0^{\eps}=x_{0},\\
\D Y_t^{\eps} & = \kappa(\theta - Y_t^{\eps})\eps\D t + \sqrt{\eps}\xi\sqrt{Y_t^{\eps}} \D Z_t,
\quad & Y_0^{\eps} = y_{0},\\
\D\left\langle W,Z\right\rangle_t & = \rho \D t.
\end{array}
\end{equation*}
Clearly, $\XX^\eps = (X^\eps, Y^\eps)$ converges to the constant process $\xx_{0} = (x_{0}, y_{0})$
almost surely as~$\eps$ tends to zero.
Choosing again $h(\eps) \equiv \eps^{-\beta}$, with $\beta \in (0, 1/2)$, we then obtain that
$\eta^\eps := \eps^{\beta-1/2}(\XX^{\eps} - \xx_{0})$ satisfies a LDP as~$\eps$ tends to zero,
and that~$X^{\eps}$ satisfies a MDP
from Corollary~\ref{C:MDP_FinancialModelsOneComp} with rate function
\begin{equation*}
I(\phi) =
\left\{
\begin{array}{ll}
\displaystyle \frac{1}{2y_0}\int_{0}^{T}\dot{\phi}_s^2\D s, & \text{if }\phi \in \mathcal{AC}([0, T],\RR)
\text{ and }\phi_0 = 0,\\
+\infty, & \text{otherwise}.
\end{array}
\right.
\end{equation*}
Together with Remark~\ref{rem:Projection1},
this corresponds to the Moderate Deviations regime for the Heston model in~\cite[Theorem 6.2]{FGP}, with $\hat{\beta} = 1/2-\beta$.

\section{Large-time moderate deviations}\label{S:LargeTimeMDP}
The former section considered small-noise perturbations of stochastic It\^o diffusions,
which lent himself perfectly to small-time analysis of the solution of the system by scaling.
We investigate here proper multiscaled diffusions, in the form of~\eqref{eq:MainIntroEps},
and show how their behaviour yields large-time asymptotics of the solution.
Consider a general perturbed two-dimensional diffusion $\XX = (X,Y)$ of the form
\begin{equation}\label{eq:Main2a}
\begin{array}{rl}
\D X_t & = \displaystyle -\frac{1}{2}\sigma^2(Y_t) \D t + \sigma(Y_t)\D W_t,\\
\D Y_t & = f(Y_t) \D t + g(Y_t)\D Z_t,\\
\D\langle W, Z\rangle_t & = \rho \D t,
\end{array}
\end{equation}
with starting point $\XX_0 = (x_0, y_0)$. For $0<\eps,\delta<1$,
the general rescaling~\eqref{Eq:TimeSpaceTransformationIntro} yields
\begin{equation}\label{eq:Main2c}
\begin{array}{rl}
\D X^{\eps}_t & = \displaystyle -\frac{1}{2}\frac{\eps}{\delta}\sigma^2(Y^{\eps}_t) \D t + \sqrt{\eps}\sigma(Y^{\eps}_t)\D W_t,\\
\D Y^{\eps}_t & = \displaystyle \frac{\eps}{\delta^{2}}f(Y^{\eps}_t) \D t + \frac{\sqrt{\eps}}{\delta}g(Y^{\eps}_t)\D Z_t,\\
\D\langle W, Z\rangle_t & = \rho \D t,
\end{array}
\end{equation}
If $\delta=\eps$, then $(X^{\eps}_{t},Y^{\eps}_{t})=(\eps X_{t/\eps},Y_{t/\eps})$,
but the more general space-time transformation~\eqref{Eq:TimeSpaceTransformationIntro}
allows us to consider also the case $\frac{\eps}{\delta}\rightarrow\gamma\in[0,\infty)$.
For $0<\eps,\delta\ll 1$ this is in the setup of~\cite{G03} and more relevant,
if we allow the possibility of $\frac{\eps}{\delta}\rightarrow\gamma\in(0,\infty)$
and correlation $\rho\neq 0$, the setup of~\cite{MS}.
To this end let us make the following standing assumption:

\begin{assumption} \label{C:ergodic}\
\begin{enumerate}[(i)]
\item{The function  $\sigma(\cdot) \in \mathcal{C}^{\alpha}(\RR)$ for some $\alpha>0$ and there exist constants $0<K<\infty$, $0\leq q_{\sigma}<1$ such that
\begin{equation*}
\lvert \sigma( y) \rvert \le K (1 + \lvert y \rvert^{q_{\sigma}}).
\end{equation*}}
\item{The function $f(\cdot)$ is locally bounded and we can write $f(y)=-\kappa y+\tau(y)$ where $\tau(\cdot)$ is a globally Lipschitz function with Lipschitz constant $L_{\tau}<\kappa$. In addition, $\lim\limits_{|y|\uparrow\infty}\frac{f(y)y}{|y|^{2}} = -\kappa$.}
\item{The function $g$ is either uniformly continuous and bounded from above and away from zero or alternatively it takes the form $g(y)=\xi y^{q_{g}}$ for $q_{g} \in[1/2,1)$ with~$\xi$ a non-zero constant.
In the latter case, we also assume that $f(y)=\kappa(\theta-y)$.}
\item{The constants $q_{\sigma}$ and $q_{g}$ satisfy $q_{\sigma}+q_{g}\leq 1$.}
\end{enumerate}
\end{assumption}

\begin{assumption}\label{C:StrongSolution}
The SDE~\eqref{eq:Main2a} has a unique strong solution.
\end{assumption}

Let $\Ll_{Y}$ denote the infinitesimal generator of the $Y$ process  before time/space rescalings, i.e.
\begin{align}
  \Ll_{Y} h &=  f h^{\prime} + \frac{1}{2} g^{2} h^{\prime\prime},\label{Eq;Operator}
\end{align}
and let $\mu(\cdot)$ be the invariant measure corresponding to~$\Ll_{Y}$,
which under Assumption~\ref{C:ergodic} exists and is unique. With
\[
\lim_{\eps \downarrow 0} \frac{\eps}{\delta} =: \gamma\in(0,\infty),
\]
let us set
$$
\lambda(y) :=-\frac{1}{2}\gamma\sigma^{2}(y)
\qquad\text{and}\qquad
\overline{\lambda} := \int_{\RR}\lambda(y)\mu(\D y).
$$
Then, it is a classical result~\cite{PV01} that~$X^\eps$ converges to~$\overline{X}$ in probability
as~$\eps$ tends to zero, where
\begin{equation*} \label{E:averagedSDE}
  \overline{X}_t := x_{0}+ \overline{\lambda}t.
\end{equation*}

Lastly, we introduce the function $\Phi$, solving the Poisson equation
\begin{equation} \label{E:Phi}
  \Ll_{Y} \Phi( y) = - \frac{1}{\gamma}(\lambda(y) - \overline{\lambda} )=\frac{1}{2}\left(\sigma^{2}(y)-\overline{\sigma}^{2}\right),
  \qquad \text{with} \qquad
  \int_\mathcal{Y} \Phi(y) \mu(\D y) = 0,
\end{equation}
where $\overline{\sigma}^{2}:= \int_{\RR}\sigma^2(y)\mu(\D y)$.
In the case of bounded  and non-degenerate $g$, Assumption~\ref{C:ergodic} implies that Theorems 1 and 2 in~\cite{PV01} hold,
providing existence and appropriate smoothness of~$\Phi$ and also polynomial growth (in~$|y|$).
In the case where~$g$ is only H\"{o}lder continuous with exponent $q_{g}\in [1/2,1)$ and can become zero only at zero, the same conclusions hold by Lemma~\ref{T:regularityCIR}.
In order to state Theorem~\ref{T:LongTimeMDP}, we need to know the relative rates at which $\delta$, $\eps$, and $1 / h(\eps)$ vanish. In particular,
\begin{equation}\label{Eq:LimitingConstants}
  \zeta := \lim_{\eps \downarrow 0} \frac{\eps / \delta - \gamma}{\sqrt{\eps} h(\eps) }
\end{equation}
specifies the relative rate at which $\eps/\delta$ converges and $h(\eps)$  goes to infinity.
In order for a MDP to hold, we require that $\zeta$ be finite.
In the special case $\frac{\eps}{\delta}=\gamma$, then $\zeta=0$.
The following holds:
\begin{theorem}[Theorem 2.1 of~\cite{MS}] \label{T:LongTimeMDP}
Under Assumptions~\ref{C:ergodic} and~\ref{C:StrongSolution},
the sequence $\{X^{\eps},\eps>0\}$ from~\eqref{eq:Main2c} satisfies the MDP
with speed~$h^{2}(\eps)$ and rate function
\begin{equation*}
I(\phi) =
\inf\left\{
\frac{1}{2}\int_{0}^{T}|u_{s}|^{2}\D s:
u\in L^{2}([0,T];\RR),
\phi=\int_{0}^{\cdot}\left(\alpha+q^{1/2}u_{s}\right)\D s\right\},
\end{equation*}
where
$$
  \alpha=-\frac{\zeta}{2} \int_\RR \sigma^{2}(y) \mu (\D y)
  \qquad\text{and}\qquad
  q
  =\int_\RR \left[ \sigma^{2}(y) +  \left|\Phi^{\prime}(y) g(y)\right|^{2} + 2\rho\sigma(y) g(y)\Phi^{\prime}(y) \right] \mu(\D y),
$$
where the finite constant $\zeta$ is defined in~\eqref{Eq:LimitingConstants},
and $\Phi$ is the solution to the Poisson equation~\eqref{E:Phi}.
\end{theorem}

A  few comments on how Theorem \ref{T:LongTimeMDP} compares to Theorem 2.1 of \cite{MS} are in line.
\begin{remark}
Theorem \ref{T:LongTimeMDP} corresponds to the setting of Regime 2 in \cite{MS} and in the special case where the drift of the fast motion, $Y^{\eps}$ grows at most linearly in $y$, i.e. $r=1$ in the notation of \cite{MS}. In addition, we would like to note that Theorem 2.1 of \cite{MS} is proven under slightly stronger conditions on the growth of the coefficients than the ones made in the current Assumption~\ref{C:ergodic}. There are two reasons for this. First, in \cite{MS} the general  multidimensional case was considered. In contrast, in this paper we consider the one-dimensional case, which then means that we can get more precise information on the growth of the solution to Poisson equations like $(\ref{E:Phi})$, see Lemma \ref{T:regularityCIR}. In turn, these more precise growth rates lead to more relaxed conditions. Second, in the setting of (\ref{eq:Main2a}) and in contrast to the general setting of \cite{MS}, the coefficients are only functions of $y$ and not of $x$.  This then implies that the solution to the PDE $(\ref{E:Phi})$ is also only a function of $y$. Hence, a number of terms that need to be taken into account in \cite{MS}, see Appendix C there, are basically zero here (since $\Phi$ does not depend on $x$).  These two circumstances then imply that we can relax the growth conditions on the coefficients of the model. Then, the exact same methodology as in \cite{MS} yields the proof of Theorem \ref{T:LongTimeMDP}. The full argument is not repeated here due to its similarity to the argument in \cite{MS}.
\end{remark}

\begin{example}[Heston model]\label{ex:HestonLarge}
For the Heston model~\eqref{eq:Heston}, the rescalings~\eqref{Eq:TimeSpaceTransformationIntro}
yield the system
\begin{equation}\label{eq:HestonEpsLargeTime}
\begin{array}{rll}
\D X_t^{\eps} & = \displaystyle -\frac{\eps}{2\delta}Y_t^{\eps}\D t
 + \sqrt{\eps}\sqrt{Y_t^{\eps}}\D W_t, \quad & X_0^{\eps}=\delta x_{0},\\
\D Y_t^{\eps} & = \displaystyle \kappa(\theta - Y_t^{\eps})\frac{\eps}{\delta^{2}}\D t
 + \frac{\sqrt{\eps}}{\delta}\xi\sqrt{Y_t^{\eps}} \D Z_t,
\quad & Y_0^{\eps} = y_{0},\\
\D\left\langle W,Z\right\rangle_t & = \rho \D t.
\end{array}
\end{equation}
In this case, the invariant measure~$\mu$ has the Gamma density~\cite[Section 33.4]{FouqueEtall2011}
\begin{equation}\label{eq:GammaDensity}
  m(y) = \frac{(2\kappa/\xi^2)^{2\kappa \theta/\xi^2}}{\Gamma(2\kappa \theta/\xi^2)} y^{2\kappa \theta/\xi^2 - 1} \E^{-2\kappa y/ \xi^2},
  \qquad \text{for }y>0.
\end{equation}
Furthermore, it is easy to check that Assumption~\ref{C:ergodic} is satisfied,
and so is Assumption~\ref{C:StrongSolution} by~\cite[Chapter~5, Theorem~2.9]{KarShreve},
so that Theorem~\ref{T:LongTimeMDP} holds with
$$
\alpha = -\frac{\zeta\theta}{2}
\qquad\text{and}\qquad
q =\int_\RR y \left( 1 +  \left|\xi\Phi^{\prime}(y) \right|^{2} + 2\rho\xi\Phi^{\prime}(y) \right) \mu(\D y).
$$

The Poisson equation~\eqref{E:Phi} can be solved explicitly as $\Phi^{\prime}(y)=-\frac{1}{2\kappa}$,
which yields
\begin{equation}\label{eq:PoissonqHeston}
q=\theta\left(1+\frac{\xi^{2}}{4\kappa^{2}}-\frac{\rho\xi}{\kappa}\right)>0.
\end{equation}
\end{example}
\begin{remark}[Connections to large deviations]\label{rem:MDPLDPCurv}
For the Heston model in Example~\ref{ex:HestonLarge}, consider the regime $\eps=\delta$,
so that $\gamma=1$ and $\zeta=\alpha=0$.
Taking $h(\eps) = \eps^{-\beta}$ with $\beta \in (0,1/2)$ implies an LDP for
the sequence of processes $(\eta^\eps = \eps^{\beta-1/2}X_\cdot^\eps)_{\eps>0}$,
or, in terms of the original process, an LDP for
$\eps^{\beta+1/2}X_{\cdot/\eps}$.
Setting $t=1$ and mapping~$\eps\mapsto t^{-1}$ yields an LDP for
$X_{t} / t^{\beta+1/2}$ as~$t$ tends to infinity.
The MDP from Theorem~\ref{T:LongTimeMDP} therefore implies that, for any $x> 0$:
$$
\lim_{t\uparrow\infty}t^{-2\beta}\log\PP\left(\frac{X_t}{t^{\beta+1/2}}\geq x\right)
 = -\inf\left\{I(\phi): \phi_{0}=0,\phi_{1}\geq x\right\}.
$$
The Euler-Lagrange equation for this minimisation problem is simply $-q\ddot{\phi}_t=0$, which yields,
with the boundary conditions, the optimal path $\phi_t = x t $.
Therefore, we obtain
$$
\lim_{t\uparrow\infty}t^{-2\beta}\log\PP\left(\frac{X_t}{t^{\beta+1/2}}\geq x\right)
 = -\frac{x^2}{2q},
$$
with~$q$ given in~\eqref{eq:PoissonqHeston}.
Now, the large-time large deviations regime was proved in~\cite{FordeJac, JKRM}, with
$$
\lim_{t\uparrow\infty}t^{-1}\log\PP\left(t^{-1}X_t\geq x\right)
 = -\Lambda^*(x),
$$
and rate function~$\Lambda^{*}$ available in closed form as
$$
\Lambda^*(x):=u^*(x)x-\Lambda(u^*(x)),\quad\text{for all }x\in\RR,
$$
with
\begin{align*}
\Lambda(u) & := \frac{\kappa\theta}{\xi^2}\Big(\kappa-\rho\xi u-d(u)\Big)
\qquad\text{and}\qquad
d(u) := \sqrt{(\kappa-\rho\xi u)^2 + \xi^2 u\left(1-u^2\right)},\quad\text{for all }u\in\left(u_-,u_+\right),\\
u^*(x) & := \frac{\xi-2\kappa\rho+(\kappa\theta\rho+x\xi)\widehat{\zeta}\left(x^2\xi^2+2x\kappa\theta\rho\xi+\kappa^2\theta^2\right)^{-1/2}}{2\xi(1-\rho^2)},
\quad
u_{\pm}:=\frac{\xi-2 \kappa \rho \pm \widehat{\zeta}}{2\xi(1-\rho^2)},
\end{align*}
and $\widehat{\zeta} := \sqrt{4\kappa^2+\xi^2-4\kappa\rho\xi}$.
Here again, the moderate rate function turns out to be the second-order Taylor expansion
of the large deviations rate function at its minimum.
More precisely, it is easy to show that~$\Lambda^*$ attains its minimum at~$-\theta/2$,
which corresponds to the limiting behaviour (at $t=1$) of $X_1^{\eps}$ as~$\eps$ tends to zero.
Indeed, due to ergodicity, we have that $\D \overline{X}_t = -\frac{1}{2}\theta\D t$.
We can also show that
$\partial_{xx}^2\Lambda^*(-\theta/2) = q^{-1}$,
where the constant~$q$ is given by the moderate deviations regime~\eqref{eq:PoissonqHeston}.
Hence, the moderate deviations rate function  characterizes the local curvature of the large deviations rate function around its minimum.

\begin{figure}[ht]
\includegraphics[scale=0.4]{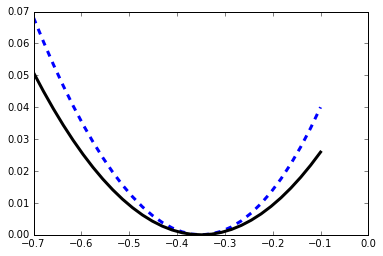}
\caption{Comparison of the MDP rate function (black, solid) and the LDP rate function (blue, dash) in the Heston model.
The MDP rate function has been shifted so that the minima match.
The parameters are
$\kappa, \theta, \sigma, \rho = 2.7609, 0.7, -0.6, 0.8$.
}
\end{figure}
\end{remark}
\section{Limiting behaviour for integrated functionals}\label{S:MDPIntegratedFunctionals}
We now consider an extension of the previous results to functionals of diffusion processes,
and prove moderate deviations thereof.
This problem has already been approached in~\cite{GuillinLipster2005}, but, using weak convergence techniques,
we are able to relax some of the their assumptions, making the results amenable to applications in mathematical finance, see also Remark \ref{R:LispterResultRemark}. We consider the model
\begin{equation}\label{eq:Main3c}
\begin{array}{rl}
\D X^{\eps}_t & = \displaystyle -\frac{1}{2}\frac{\eps}{\delta}\sigma^2(X^{\eps}_t,Y^{\eps}_t) \D t + \sqrt{\eps}\sigma(X^{\eps}_t,Y^{\eps}_t)\D W_t,\\
\D Y^{\eps}_t & = \displaystyle \frac{\eps}{\delta^{2}}f(X^{\eps}_t,Y^{\eps}_t) \D t + \frac{\sqrt{\eps}}{\delta}g(X^{\eps}_t,Y^{\eps}_t)\D Z_t,\\
\D\langle W, Z\rangle_t & = \rho \D t,
\end{array}
\end{equation}
which corresponds to the original system~\eqref{eq:MainIntro} using the general rescaling~\eqref{Eq:TimeSpaceTransformationIntro}
(as in the large-time framework above), whenever the coefficients do not depend on~$X^\eps$.
For given sets $A,B$, for $i,j\in\mathbb{N}$ and $\alpha>0$ we denote by $\mathcal{C}_b^{i, j + \alpha}(A \times B)$, the space of functions with $i$ bounded derivatives in $x$ and $j$ derivatives in $y$, with all partial derivatives being $\alpha$-H\"{o}lder continuous with respect to $y$, uniformly in $x$.
For any $x\in\RR$, let $\Ll^{Y}_{x}$ denote the infinitesimal generator of the $Y$ process  when $\epsilon=\delta=1$ (i.e. before the space/time rescalings):
\begin{align}
  \Ll^{Y}_{x} h(y) &=  f(x,y) h^{\prime}(y) + \frac{1}{2} g^{2}(x,y)h^{\prime\prime}(y),\label{Eq;Operator1}
\end{align}
and let $\mu_{x}(\D y)$ be the invariant measure corresponding to the infinitesimal generator $\Ll^{Y}_{x}$, which is guaranteed to exist by Assumption~\ref{C:ergodic1} below. With
\begin{equation}\label{eq:LimitEpsOverDelta}
\lim_{\eps \downarrow 0} \frac{\eps}{\delta}=\gamma\in(0,\infty),
\end{equation}
let us set
$$
\lambda(x,y) :=-\frac{1}{2}\gamma\sigma^{2}(x,y)
\qquad\text{and}\qquad
\overline{\lambda}(x) := \int_{\RR}\lambda(x,y)\mu_{x}(\D y).
$$

It is a classical result~\cite{PV01} that $X^\eps$ converges in probability to~$\overline{X}$ as~$\eps$ tends to zero, where
\begin{equation} \label{E:averagedSDE1}
  \overline{X}_t := x_{0}+ \int_{0}^{t}\overline{\lambda}(\overline{X}_{s})\D s.
\end{equation}

For an appropriate function $H(x,y)$, let $\overline{H}(x):=\int_{\RR} H(x,y)\mu(\D y)$;
 we are interested in the LDP for
\begin{equation}\label{Eq:IntergatedProcess}
R^{\eps} := \frac{1}{\sqrt{\eps}h(\eps)}\int_{0}^{\cdot}\left(H(X^{\eps}_{s},Y^{\eps}_{s})-\overline{H}(X^{\eps}_{s})\right)\D s.
\end{equation}
The following assumption is a counterpart to Assumption~\ref{C:ergodic} in the more general
configuration with full dependence on~$x$ and~$y$ in the coefficients.

\begin{assumption} \label{C:ergodic1}\
\begin{enumerate}[(i)]
\item{There is $\alpha>0$ such that 
either $f,g\in \mathcal{C}^{1,\alpha}(\RR)$ and $H \in \mathcal{C}^{2,1+\alpha}(\RR)$,
or $f,g\in \mathcal{C}^{1,2+\alpha}(\RR)$ and $H \in \mathcal{C}^{2,\alpha}(\RR)$.}
\item{ There exist constants $0<K<\infty$ and $0\leq q_{\sigma}<1$ such that
\begin{equation*}
    \lvert \sigma(x, y) \rvert \le K (1 + \lvert y \rvert^{q_{\sigma}}).
\end{equation*}
In addition, $\sigma$ is Lipshitz continuous in $x$ locally uniformly with respect to $y$.}
\item{We can write $f(x,y)=-\kappa y+\tau(x,y)$ with $\kappa>0$, $\tau(x,y)$ is a globally Lipschitz function in $y$, uniformly with respect to $x$, with Lipschitz constant $L_{\tau}<\kappa$,
and $\sup_{x\in\RR}f(x,y)y\leq -\kappa|y|^{2} $ for $|y|$ large enough.}
\item{The function $g$ is either uniformly continuous and bounded from above and away from zero or alternatively it takes the form $g(y)=\xi y^{q_{g}}$ for $q_{g} \in[1/2,1)$ and $\xi$ a non-zero constant. In the latter case we also assume that $f(x,y)=\kappa(\theta-y)$.
}
\item{There exist constants $0<K<\infty$ and $q_{H} \geq 0$ such that
\[
\lvert H(x, y) \rvert +\lvert \partial_{x}H(x, y) \rvert + \lvert \partial_{x}^{2}H(x, y) \rvert \le K (1 + \lvert y \rvert^{q_{H}}).
\]
}
\end{enumerate}
\end{assumption}

\begin{assumption}\label{C:StrongSolution1}
The SDE~\eqref{eq:Main3c} has a unique strong solution.
\end{assumption}

Let us further consider $u(x,\cdot)$ to be the solution to the Poisson equation
\begin{equation}\label{Eq:uPoisson}
\Ll^{Y}_{x}u(x,y)=H(x,y)-\overline{H}(x), \quad \int_{\mathcal{Y}}u(x,y)\mu_{x}(dy)=0.
\end{equation}

By Theorem~\ref{T:regularity} and Lemma~\ref{T:regularityCIR}, $u$ is a well-defined classical solution that can grow at most polynomially in~$|y|$.
Next, due to compactness issues, that we will see in the proof of tightness in Section~\ref{SS:Tightness}, we need to restrict the values of the constants $q_{\sigma},q_{g}, q_{H}$ that dictate the growth in~$|y|$ of the coefficients~$\sigma$, $g$ and~$H$ respectively.
We collect these constraints in the following assumption:
\begin{assumption}\label{C:GrowthAssumptions}\
\begin{enumerate}[(i)]
\item{If the operator $\Ll^{Y}_{x}$ given by~\eqref{Eq;Operator1} depends generally on $(x,y)$, we assume that
\[
\max\{q_{\sigma}+q_{H}, q_{g}+q_{H}\}<1.
\]}
\item{If  the operator in~\eqref{Eq;Operator1} takes the special form
\begin{align}\label{Eq:CIR}
  \Ll^{Y} h(y) &=  \kappa(\theta-y) h^{\prime}(y) + \frac{1}{2} \xi^{2} y^{2q_{g}} h^{\prime\prime}(y),
\end{align}
with $q_{g}\in[1/2,1)$, and $H(x,y)=H(y)$ is only a function of $y$, then we assume that
$$
 q_{\sigma}<1
 \qquad\text{and}\qquad
 q_{g}+q_{H}<2,
$$
where the constants $q_{g}, q_{H}$ and the function $h(\cdot)$ are such that
$\displaystyle \lim_{\eps\downarrow 0}\sqrt{\eps}h(\eps)^{\frac{q_{g}+q_{H}-1}{1-q_{g}}}=0$.
}
\end{enumerate}
\end{assumption}
Notice that for Assumption~\ref{C:GrowthAssumptions}(ii),
if $q_{H}=1$, then assuming $h(\eps) = \eps^{-\beta}$ with $\beta \in (0,1/2)$
yields $\sqrt{\eps}h(\eps)^{q_g/(1-q_g)} = \eps^{1/2-\beta q_g / (1-q_g)}$,
so that Assumption~\ref{C:GrowthAssumptions} is satisfied if and only if
$q_g < 1/(2\beta+1) \in (1/2,1)$.
In order to state the main result, we shall need one additional assumption,
merely for technical reasons.
Given the solution~$u(\cdot,\cdot)$ to the Poisson equation~\eqref{Eq:uPoisson}
and the constant~$\gamma$ defined in~\eqref{eq:LimitEpsOverDelta}, introduce the following:	
\begin{align}\label{Eq:BarQDef}
 \overline{Q}(x) :=\int_{\RR}Q(x,y)\mu(\D y)
 \qquad\text{and}\qquad
 Q(x,y) := \frac{1}{\gamma^{2}}|u_{y}(x,y)g(y)|^{2}.
\end{align}

\begin{assumption}\label{C:Positiveness}
The Lebesgue measure of the set $\left\{s\in[0,T]: \overline{Q}(\overline{X}_{s})=0\right\}$ is equal to zero.
\end{assumption}

At this point, we mention that the reason we single out the operator~\eqref{Eq:CIR}
is because in this case we can have more detailed information on the behaviour of the solutions to the corresponding Poisson equation~\eqref{Eq:uPoisson}.
This is discussed in detail in Lemma~\ref{T:regularityCIR}.
To be specific, in the special case of~\eqref{Eq:CIR} with $H(x,y)\equiv H(y)$,
the solution to the PDE~\eqref{Eq:uPoisson} grows like $|y|^{q_{H}}$ whereas the derivative grows like $|y|^{q_{H}-1}$ for $q_{H}\geq 1$. On the other hand, in the general case, one can only guarantee that both solution and derivatives to the solution of PDE~\eqref{Eq:uPoisson}
grow like $|y|^{q_{H}}$.
Then, we have the following result, proved in Sections~\ref{S:ControlRep}-\ref{S:ProofIntegratedFunctMDP}:
\begin{theorem} \label{T:IntegratedFunctionsl}
Let Assumptions~\ref{C:ergodic1},~\ref{C:StrongSolution1}, \ref{C:GrowthAssumptions} and~\ref{C:Positiveness} be satisfied.
Then, the sequence $\{R^{\eps},\eps>0\}$ from~\eqref{Eq:IntergatedProcess} satisfies the LDP (equivalently MDP for $\int_{0}^{\cdot}\left(H(X^{\eps}_{s},Y^{\eps}_{s})-\overline{H}(X^{\eps}_{s})\right)\D s$), with speed $h^{2}(\eps)$ and  rate function
\begin{equation*}
I(\phi) =
\inf\left\{
\frac{1}{2}\int_{0}^{T}|v_{s}|^{2}\D s:
v\in L^{2}([0,T];\RR),
\phi_{\cdot} =  \int_{0}^{\cdot} \overline{Q}^{1/2}(\overline{X}_s) v_{s}\D s\right\}.
\end{equation*}
\end{theorem}

\begin{remark}\label{R:LispterResultRemark}
We point out that a result similar to Theorem~\ref{T:IntegratedFunctionsl}
has also been established in~\cite{GuillinLipster2005} using different methods.
However, the results of~\cite{GuillinLipster2005} are not sufficient for our purposes,
as they assumed that both functions~$g$ and~$H$ are bounded
($H$ is even assumed to be bounded by $K(1+|y|^{-(2+\eta)})$ for some $\eta>0$).
In our work we allow the coefficients to grow according to Assumptions~\ref{C:ergodic1}
and~\ref{C:GrowthAssumptions}.
In turn this allows us to consider a considerably wider class of models. We see some examples below.
\end{remark}

\section{Financial applications: asymptotic behaviour of option prices}\label{S:FinancialApplications}\
We can use Corollary~\ref{C:MDP_FinancialModelsOneComp} for the short-time asymptotics and Theorem~\ref{T:LongTimeMDP} for the long-time asymptotics to get statements for Call option prices analogous to Theorem 6.2 and 6.3 of~\cite{FGP}. Theorem \ref{T:IntegratedFunctionsl} is used to obtain estimates on asymptotics for realised variance.

\subsection{Tail estimates}
We show here how the general result in Theorem~\ref{T:GeneralResult},
as well as its lighter versions in Proposition~\ref{P:MDP_FinancialModelsTwoComp}
and Corollary~\ref{C:MDP_FinancialModelsOneComp} yield, for some suitable scaling, tail estimates
for probabilities.
We shall consider the following forms for the coefficients of the SDE~\eqref{eq:Main}:
\begin{assumption}\label{assu:TailScaling}
There exist $(a, b, c_g, c_\sigma) \in\RR^4$,
together with $\nu_{\sigma} \in (0,1]$ and $\nu_{g} \in [0,1-\nu_{\sigma}]$ such that
$$
f(t, x, y) = a+by,
\qquad
g(x, y) = c_g y^{\nu_g},
\qquad
\sigma(x, y) = c_\sigma y^{\nu_\sigma},
\qquad
\text{for all }t, x, y.
$$
\end{assumption}
These assumptions may look restrictive, but are in fact sufficient for many financial applications,
as detailed further below.
Note that we excluded the trivial case $\nu_\sigma=0$ as the SDE for~$X$ is then independent of~$Y$,
and can be dealt with directly by hand.

\begin{proposition}\label{prop:Tail}
Consider the model~\eqref{eq:Main} under Assumptions~\ref{A:MainGeneralAssumption}
and~\ref{assu:TailScaling}.
Then, with $\zeta:= \frac{1-\nu_{g}+\nu_{\sigma}}{2(1-\nu_{g})}$,
the sequence $(\eps^{\zeta} X)$ satisfies a moderate deviations principle
on $\Cc([0,T],\RR)$ with speed~$h(\eps)^2$ and rate function
 \begin{equation*}
I(\phi) =
\left\{\begin{array}{ll}
\displaystyle \frac{1}{2\sigma(y_{0})^2} \int_0^T |\dot{\phi}_s|^{2}\D s,
 & \text{if }\phi \in \mathcal{AC}([0, T],\RR) \text{ and } \phi_{0}=0,\\
+\infty, & \text{otherwise}.
\end{array}
\right.
\end{equation*}
\end{proposition}
\begin{proof}
Let $\alpha>0$ to be chosen.
Under Assumption~\ref{assu:TailScaling}, the model~\eqref{eq:Main} reads
\begin{equation*}
\begin{array}{rl}
\D X_t & = \displaystyle -\frac{c_\sigma^2}{2} Y_t^{2\nu_\sigma} \D t + c_\sigma Y_t^{\nu_\sigma}\D W_t,\\
\D Y_t & = (a+bY_t) \D t + c_g Y_t^{\nu_g}\D Z_t,\\
\D\langle W, Z\rangle_t & = \rho \D t,
\end{array}
\end{equation*}
with starting point $\XX_0 = (x_0, y_0)$.
The rescaling
$\XX^\delta_t = (X^\delta_t, Y^\delta_t) := (\delta X_t, \delta^{\alpha} Y_t)$ for all $t\geq 0$ yields the system
\begin{equation*}
\begin{array}{rl}
\D X_t^\delta & = \displaystyle -\frac{\delta^{1-2\alpha\nu_{\sigma}}c_\sigma^2}{2}(Y_t^\delta)^{2\nu_{\sigma}} \D t
 + c_\sigma\delta^{1-\alpha\nu_{\sigma}}(Y_t^{\delta})^{\nu_{\sigma}}\D W_t,\\
\D Y_t^\delta & = \left(a\delta^\alpha + b Y_t^{\delta}\right) \D t
 + c_g\delta^{\alpha(1-\nu_{g})}(Y_t^{\delta})^{\nu_{g}}\D Z_t,\\
\D\langle W, Z\rangle_t & = \rho \D t,
\end{array}
\end{equation*}
with starting point $\XX_0^\delta = (X_0^\delta, Y_0^\delta) = (\delta x_0, \delta^{\alpha}y_0)$.
Equating the~$\delta$ powers in the diffusion coefficients yields $1-\alpha\nu_{\sigma} = \alpha(1-\nu_{g})$, or
$\alpha = (1-\nu_{g}+\nu_{\sigma})^{-1}>0$,
so that
\begin{equation*}
\begin{array}{rl}
\D X_t^\delta & = \displaystyle -\frac{\delta^{1-2\alpha\nu_{\sigma}}c_\sigma^2}{2}(Y_t^\delta)^{2\nu_{\sigma}} \D t
 + c_\sigma\delta^{\gamma}(Y_t^{\delta})^{\nu_{\sigma}}\D W_t,\\
\D Y_t^\delta & = \left(a\delta^\alpha + b Y_t^{\delta}\right) \D t
 + c_g\delta^{\gamma}(Y_t^{\delta})^{\nu_{g}}\D Z_t,\\
\D\langle W, Z\rangle_t & = \rho \D t,
\end{array}
\end{equation*}
with $\gamma:=1-\alpha\nu_{\sigma}$.
Setting $\sqrt{\eps} := \delta^\gamma$ finally gives the system
\begin{equation*}
\begin{array}{rl}
\D X_t^\eps & = \displaystyle -\frac{\eps^{(1-2\alpha\nu_{\sigma})/(2\gamma)}c_\sigma^2}{2}(Y_t^\eps)^{2\nu_{\sigma}} \D t
 + c_\sigma\sqrt{\eps}(Y_t^{\eps})^{\nu_{\sigma}}\D W_t,\\
\D Y_t^\eps & = \left(a\eps^{\frac{1}{2(1-\nu_{g})}} + b Y_t^{\eps}\right) \D t
 + c_g\sqrt{\eps}(Y_t^{\eps})^{\nu_{g}}\D Z_t,\\
\D\langle W, Z\rangle_t & = \rho \D t,
\end{array}
\end{equation*}
with starting point $\XX_0^\eps = (\eps^{1/(2\gamma)} x_0, \eps^{\alpha/(2\gamma)}y_0)$.
Theorem~\ref{T:GeneralResult} then applies whenever~$\gamma>0$ and $1-2\nu_{\sigma}\alpha\geq 0$,
which is equivalent to $\nu_{g}\leq 1-\nu_{\sigma}$,
and the proposition follows by setting $\zeta:=1/(2\gamma)$.
\end{proof}

\begin{example}\
\begin{itemize}
\item In the Heston case~\eqref{eq:Heston},
Assumption~\ref{assu:TailScaling} holds with $a=\kappa\theta$, $b=-\kappa$, $c_\sigma=1$, $c_g=\xi$,
and $\nu_{\sigma}=1/2=\nu_{g}$,
so that Proposition~\ref{prop:Tail} applies with $\zeta=1$;
\item In the Stein-Stein case~\eqref{eq:SteinStein},
Assumption~\ref{assu:TailScaling} holds with $a=a$, $b=b$, $c_\sigma=1$, $c_g=c$, $\nu_{\sigma}=1$,
$\nu_{g}=0$, so that
Proposition~\ref{prop:Tail} applies with $\zeta=1$.
\end{itemize}
In both cases, for any $x>x_0$, we can write
$$
\lim_{\eps\downarrow 0}\frac{1}{h^2(\eps)}\log\PP\left(\frac{X_{t}^{\eps}}{\sqrt{\eps}h(\eps)} \geq x\right)
 = \lim_{\eps\downarrow 0}\frac{1}{h^2(\eps)}\log\PP\left(X_{t} \geq \frac{xh(\eps)}{\sqrt{\eps}}\right)
 = -\frac{x^2}{2\sigma(y_0)^2 t}.
$$
Again, the moderate deviations regime provides a closed-form representation
for the rate function, in contrast with the otherwise more abstract rate function
arising from large deviations, as developed in~\cite{DFJV1, DFJV2, GuliStein}.
\end{example}

\subsection{Small-time behaviour}\label{SS:SmallTimeExamples}
The following theorem is a reformulation of~\cite[Theorem 6.3]{FGP} and~\cite[Corollary 2.1 and Theorem 2.2]{FFK12},
based on the asymptotic behaviour (assuming $x_0=0$ for simplicity)
\begin{equation}\label{eq:FGPProba}
\PP(X_t\geq k_t) \sim \exp\left\{-\frac{k_t^2}{2 \sigma^{2}(0,y_{0}) t}\right\},
\end{equation}
as~$t$ tends to zero,
from Proposition~\ref{P:MDP_FinancialModelsTwoComp} and Corollary~\ref{C:MDP_FinancialModelsOneComp},
with $k_t := k \sqrt{t}h(t)>0$.
This behaviour in fact follows directly from our main result above;
in the model~\eqref{eq:Main}, the pathwise moderate deviations principle for
the first component (Corollary~\ref{C:MDP_FinancialModelsOneComp})
generalises (being pathwise and with weaker assumptions) the results from~\cite{FGP},
and directly yields~\eqref{eq:FGPProba}.
More precisely, from the proposed scaling $(X_t^{\eps}, Y_t^{\eps}):=(X_{\eps t}, Y_{\eps t})$,
Corollary~\ref{C:MDP_FinancialModelsOneComp} and Remark~\ref{rem:Projection1} imply,
as~$t$ tends to zero,
for $k>x_0=0$,
$$
\lim_{\eps\downarrow 0}\frac{1}{h(\eps)^2}\log\PP\left(\frac{X_\eps}{\sqrt{\eps}h(\eps)}\geq k\right) = -\frac{k^2}{2\sigma^{2}(0,y_{0})},
$$
which in turn implies~\eqref{eq:FGPProba},
identifying~$t$ to~$\eps$ and setting $k_t := kh(t)\sqrt{t}$.
In order to state the main result here, we need the following assumption:
\begin{assumption}\label{assu:ExplosionTime}
For any $p\geq 1$ there exists $t^{*}_{p}>0$ such that
$\EE\left(\E^{pX_t}\right)$ is finite for all $t\in [0,t^*_{p}]$.
\end{assumption}
\begin{remark}
The moment assumption in the theorem is classical in the mathematical finance literature
related to large and moderate deviations for option pricing.
It can be found in~\cite[Theorem 6.3]{FGP}
and in the more general framework~\cite[Assumption (A2)]{FrizPigato}.
Essentially, this assumption ensures that some integrability of the stock price~$\E^{X_t}$
is guaranteed (otherwise
Call prices may not be defined), which is not automatically granted by large and moderate deviations results.
We refer the reader to~\cite[Lemma 4.7]{FrizPigato} for some easy-to-check condition.
\end{remark}
\begin{theorem}
Under Assumption~\ref{assu:ExplosionTime} and the assumptions of Proposition~\ref{P:MDP_FinancialModelsTwoComp}, the Call price satisfies
$$
\lim_{t\downarrow 0}\frac{1}{h^2(t)}\log\EE\left(\E^{X_t} - \E^{k_t}\right)_+
 = -\frac{k^2}{2 \sigma^{2}(0,y_{0})},
 \qquad\text{for all }k>0.
$$
\end{theorem}

\begin{proof}
We follow here the proof of~\cite[Theorem 6.3]{FGP} with minor modifications.
In the small-maturity case, the limiting process $\lim_{\eps\downarrow 0}X^\eps$
is equal to $x_0=0$, so that Corollary~\ref{C:MDP_FinancialModelsOneComp}
and Remark~\ref{rem:Projection1} imply that,
as~$\eps$ tends to zero, and identifying~$t$ to~$\eps$,
$$
\lim_{t\downarrow 0}\frac{1}{h^2(t)}\log\PP\left(\frac{X_t}{\sqrt{t}h(t)}\geq k\right)
 = -\frac{k^2}{2\sigma(0,y_0)^2}.
$$
For any $\delta>0$, there exists then $t^*>0$ such that for all $t \in (0, t^*]$,
letting $k_t := k\sqrt{t}h(t)$,
$$
 \exp\left\{-\frac{k_t^2}{2\sigma(0,y_0)^2 t}\left(1-c\delta\right)\right\}
   \leq \PP\left(X_t\geq k_t\right) \leq
   \exp\left\{-\frac{k_t^2}{2\sigma(0,y_0)^2 t}\left(1+c\delta\right)\right\},
$$
where the constant $c = 2\sigma(0,y_0)^2$ is of no importance.
Letting $\tilde{k}_t:=k_t(1+c\delta)$, the Call price reads
$$
\EE\left(\E^{X_t} - \E^{k_t}\right)_+
\geq \EE\left[\left(\E^{X_t} - \E^{k_t}\right)_+\ind_{X_t\geq k_t}\right]
\geq \EE\left[\left(\E^{X_t} - \E^{k_t}\right)_+\ind_{X_t\geq \tilde{k}_t}\right]
\geq \left(\E^{\tilde{k}_t} - \E^{k_t}\right)_+\PP\left(X_t\geq \tilde{k}_t\right).
$$
Regarding the first term, we can write
$$
\left(\E^{\tilde{k}_t} - \E^{k_t}\right)_+ = c\delta k_t + \Oo\left(k_t^2\right).
$$
The second term is dealt with from the moderate deviations principle
with~$k_t$ replaced by~$\tilde{k}_t$:
$$
\lim_{t\downarrow 0}\frac{1}{h^2(t)}\log\PP\left(X_t\geq \tilde{k}_t\right)
 \geq -\frac{k^2(1+c\delta)^2}{2\sigma(0,y_0)^2},
$$
and therefore
$$
\liminf_{t\downarrow 0}\frac{1}{h^2(t)}\log\EE\left(\E^{X_t} - \E^{k_t}\right)_+
\geq -\frac{k^2(1+c\delta)^2}{2\sigma(0,y_0)^2}.
$$
Letting~$\delta$ tends to zero yields the desired lower bound.
Regarding the upper bound, fix $p\geq1$ and note that, by Assumption~\ref{assu:ExplosionTime},
$\EE(\E^{pX_t})$ is finite for all $t\in [0,t^*_{p}]$.
Doob's inequality~\cite[Theorem 3.8]{KarShreve}, implies
$$
\PP\left(\overline{X}_{t^*_{p}} \geq k\right) \leq \E^{-pk}\EE\left[\exp\left(X_{t^*_{p}}^{p}\right)\right],
$$
where $\overline{X}_t := \sup_{0\leq u\leq t}X_u$, so that
$\overline{X}_{t^*_{p}}$ also has finite $p$-th moment, and
$$
\EE\left[\exp\left(X_{t^*_{p}}^{p}\right)\right]
 = p\int_{0}^{\infty}\E^{(p-1)k}\PP\left(\overline{X}_{t^*_{p}} \geq k\right)\D k
 \text{ is finite}.
$$
By dominated convergence and since the process~$X$ has continuous paths, then
$\EE(\E^{pX_t})$ converges to~$\E^{pX_0}$ as~$t$ tends to zero.
Finally, for any $p, q \in [1,\infty)$ such that $p^{-1} + q^{-1} = 1$, H\"older's inequality yields
$$
\EE\left(\E^{X_t} - \E^{k_t}\right)_+
  = \EE\left[\left(\E^{X_t} - \E^{k_t}\right)_+\ind_{\{X_t\geq k_t\}}\right]
 \leq \EE\left[\left\{\left(\E^{X_t} - \E^{k_t}\right)_+\right\}^p\right]^{\frac{1}{p}}
 \PP(X_t\geq k_t)^{\frac{1}{q}}
  \leq \EE\left(\E^{pX_t}\right)^{\frac{1}{p}}
 \PP(X_t\geq k_t)^{\frac{1}{q}}.
$$
Since the first term converges to $\EE\left(\E^{X_0}\right)$, we then obtain
$$
\limsup_{t\downarrow 0} \frac{1}{h^2(t)}\log\EE\left(\E^{X_t} - \E^{k_t}\right)_+
\leq \frac{1}{q}\limsup_{t\downarrow 0} \frac{1}{h^2(t)}\log\PP(X_t\geq k_t)
 = -\frac{k^2}{2 \sigma^{2}(0,y_{0})}.
$$
Letting~$p$ tend to infinity, and consequently~$q$ to one, yields the lower bound in the theorem.
\end{proof}

\subsection{Large-time behaviour}\label{SS:LargeTimeExamplesFinancial}
Following the methodology developed in~\cite{FordeJac, JKRM}, we can translate the moderate deviations results into large-time asymptotic behaviours of option prices.
In order to state the results,
let~$J$ denote the (convex) contraction of the rate function~$I$ given in Theorem~\ref{T:LongTimeMDP} onto the last point, i.e. $J(x):=\inf\{I(\phi):\phi(0)=0,\phi(1)=x\}$.
Introduce further the Share measure
$\QQ(A):=\EE^{\PP}\left(\E^{X_t}\ind_{\{A\}}\right)$ for any $A\in\Ff_t$.
Before proving the main result of this section, let us state and prove a useful lemma.
\begin{lemma}\label{lem:MDPShare}
Under~$\QQ$, the process~$X$ defined in~\eqref{eq:Main2a} satisfies
\begin{equation}\label{eq:SDEUnderQ}
\begin{array}{rll}
\D X_t & = \displaystyle \frac{1}{2}\sigma^2(Y_t) \D t + \sigma(Y_t)\D W^{\QQ}_t, & X_0=x_0,\\
\D Y_t & = f^{\QQ}(Y_t) \D t + g(Y_t)\D Z^{\QQ}_t, & Y_0=y_0,\\
\D\langle W^{\QQ}, Z^{\QQ}\rangle_t & = \rho \D t,
\end{array}
\end{equation}
where $f^{\QQ}(\cdot) := f(\cdot) + \rho g(\cdot)\sigma(\cdot)$.
Let Assumptions~\ref{C:ergodic} and~\ref{C:StrongSolution} be satisfied for~\eqref{eq:SDEUnderQ}.
With the rescaling~\eqref{Eq:TimeSpaceTransformationIntro},
the sequence~$(X^{\eps})_{\eps>0}$ satisfies a MDP under~$\QQ$
with speed~$h^{2}(\eps)$ and action functional~$I^{\QQ}$
defined as in Theorem~\ref{T:LongTimeMDP} with
$\overline{\lambda}^{\QQ} = -\overline{\lambda}$,
$\lambda^{\QQ}(\cdot) = -\lambda(\cdot)$
and where the invariant measure $\mu^{\QQ}$ is that of~$Y$ in~\eqref{eq:SDEUnderQ}.
\end{lemma}
\begin{proof}
From~\eqref{eq:Main2a}, we can write
$$
\QQ(A) = \EE^{\PP}\left[\exp\left\{
-\frac{1}{2}\int_{0}^{t}\sigma(Y_u)^2 \D u +\rho \int_{0}^{t}\sigma(Y_u)\D Z_u
 + \rrho\int_{0}^{t}\sigma(Y_u)\D Z_u^{\top}
\right\}\ind_{\{A\}}\right],
$$
where we decomposed the Brownian motion~$W$ into $W = \rho Z + \bar{\rho} Z^{\top}$.
Girsanov's Theorem implies that the two processes~$Z^{\QQ}$ and~$\overline{Z}^{\QQ}$ defined as
$$
Z^{\QQ}_t := Z_t -\rho \int_{0}^{t}\sigma(Y_u)\D u
\qquad\text{and}\qquad
\overline{Z}^{\QQ}_t := Z_t^{\top} -\rrho \int_{0}^{t}\sigma(Y_u)\D u
$$
are two independent standard Brownian motions under~$\QQ$, and we can rewrite~\eqref{eq:Main2a}
as~\eqref{eq:SDEUnderQ} with $f^{\QQ}(\cdot) := f(\cdot) + \rho g(\cdot)\sigma(\cdot)$
and $W^\QQ := \rho Z^{\QQ} + \rrho \overline{Z}^{\QQ}$.
The lemma then follows directly from Theorem~\ref{T:LongTimeMDP} under the assumptions on the coefficients.
\end{proof}

Note the flipped sign in the drift of~$X$ in~\eqref{eq:SDEUnderQ}.
We can thus define $J^{\QQ}$ as the contraction (onto the last point) of the rate function~$I^{\QQ}$
for~$X^\eps$ under~$\QQ$.
The minimising functional~$\phi$ in Theorem~~\ref{T:LongTimeMDP}
is linear, of the form $\phi_t = \alpha t $, hence
\begin{equation}\label{eq:simpleRate}
J(x) = \frac{x^{2} T}{2q}
\qquad\text{and}\qquad
J^{\QQ}(x) = \frac{x^{2} T}{2q^{\QQ}},
\qquad \text{for any }x \in\RR.
\end{equation}

\begin{proposition}\label{prop:LargeTOption}
Let $\beta \in (0,1/2)$.
As~$t$ tends to infinity, we observe the following asymptotic behaviours:
\begin{equation*}
\begin{array}{rll}
\displaystyle \frac{1}{t^{1/2+\beta}}\log\EE^{\PP}\left(\E^{x t^{\beta +1/2}} - \E^{X_t}\right)_+
& \approx \left\{
\begin{array}{ll}
\displaystyle x  - t^{\beta-1/2}J(x), \\
\displaystyle x,
\end{array}
\right.
 & \left.
\begin{array}{ll}
\text{if }x<0,\\
\text{if }x\geq 0,
\end{array}
\right.
\\
\displaystyle \lim_{t\uparrow\infty}\frac{1}{t^{2\beta}}\log\EE^{\PP}\left(\E^{X_t} - \E^{x t^{\beta +1/2}}\right)_+
 &  = \left\{
 \begin{array}{ll}
-J^{\QQ}(x), \\
0,
\end{array}
\right.
 & \left.
\begin{array}{ll}
\text{if }x>0,\\
\text{if }x\leq 0.
\end{array}
\right.
\end{array}
\end{equation*}
\end{proposition}
\begin{proof}
We follow the methodology developed in~\cite[Theorem 13]{JKRM} with a few changes.
For all $t\geq 1$ and $\eps>0$, the inequalities
$$
\E^{x t^{\beta +1/2}} \left(1-\E^{-\eps}\right)\ind_{\{X_t/t^{\beta +1/2} < x-\eps\}}
\leq \left(\E^{x t^{\beta +1/2}} - \E^{X_t}\right)_+
\leq \E^{x t^{\beta +1/2}} \ind_{\{X_t/t^{\beta +1/2} < x\}}
$$
hold.
Taking expectations, logarithms and dividing by $t^{2\beta}$, we obtain
$$
\frac{x}{t^{\beta-1/2}} + \frac{\log\left(1-\E^{-\eps}\right)}{t^{2\beta}}
+\frac{\log\PP\left(\frac{X_t}{t^{\beta +1/2}} < x-\eps\right)}{t^{2\beta}}
\leq \frac{1}{t^{2\beta}}\log\EE\left(\E^{x t^{\beta +1/2}} - \E^{X_t}\right)_+
\leq \frac{x}{t^{\beta-1/2}} + \frac{\log\PP\left(\frac{X_t}{t^{\beta +1/2}} < x\right)}{t^{2\beta}}.
$$
From the moderate deviations principle in Theorem~\ref{T:LongTimeMDP} and using~\eqref{eq:simpleRate},
we can write
\begin{equation*}
\lim_{t\uparrow\infty}t^{-2\beta}\log\PP\left(\frac{X_t}{t^{\beta +1/2}} < x\right) = -\inf_{z<x}J(z) =
\left\{
\begin{array}{rl}
-J(x), & \text{if }x<0,\\
0, & \text{if }x \geq 0,
\end{array}
\right.
\end{equation*}
so that, as~$t$ tends to infinity, the asymptotic behaviour for Put option prices in the proposition holds.

For Call option prices on~$\exp(X)$, we can use the inequalities, for $t\geq 1$ and $\eps>0$,
$$
\E^{X_t} \left(1-\E^{-\eps}\right)\ind_{\{X_t/t^{\beta +1/2} > x+\eps\}}
\leq \left(\E^{X_t} - \E^{x t^{\beta +1/2}}\right)_+
\leq \E^{X_t} \ind_{\{X_t/t^{\beta +1/2} > x\}}.
$$
Taking expectations under~$\PP$, this becomes
$$
\left(1-\E^{-\eps}\right) \EE^{\PP}\left(\E^{X_t} \ind_{\{X_t/t^{\beta +1/2} > x+\eps\}}\right)
\leq \EE^{\PP}\left(\E^{X_t} - \E^{x t^{\beta +1/2}}\right)_+
\leq \EE^{\PP}\left(\E^{X_t} \ind_{\{X_t/t^{\beta +1/2} > x\}}\right),
$$
and, using the Share measure~$\QQ$, this translates into
$$
\left(1-\E^{-\eps}\right) \QQ\left(\frac{X_t}{t^{\beta +1/2}} > x+\eps\right)
\leq \EE^{\PP}\left(\E^{X_t} - \E^{x t^{\beta +1/2}}\right)_+
\leq \QQ\left(\frac{X_t}{t^{\beta +1/2}} > x\right).
$$
Using Lemma~\ref{lem:MDPShare} and~\eqref{eq:simpleRate},
the proposition then follows from the computation
\begin{equation*}
\lim_{t\uparrow\infty}t^{-2\beta}\QQ\left(\frac{X_t}{t^{\beta +1/2}} > x\right)
 = -\inf_{z>x}J^{\QQ}(z) =
\left\{
\begin{array}{rl}
-J^{\QQ}(x), & \text{if }x>0,\\
0, & \text{if }x \leq 0.
\end{array}
\right.
\end{equation*}
\end{proof}

\subsection{Asymptotics of the realised variance}\label{SS:AsymptoticsRealizedVariance}
We now show how Theorem~\ref{T:IntegratedFunctionsl}
applies to the realised variance in the Heston model~\eqref{eq:Heston}.
The rescaling~\eqref{Eq:TimeSpaceTransformationIntro} yields
the same perturbed system ~\eqref{eq:HestonEpsLargeTime} as in the large-time case.
Likewise, the invariant measure $\mu(\cdot)$ has Gamma density given in~\eqref{eq:GammaDensity}.
With $H(x,y)\equiv y$, we have $q_{\sigma}=q_{g}=1/2$ and $q_{H}=1$,
so that Assumption~\ref{C:GrowthAssumptions} holds and Theorem~\ref{T:IntegratedFunctionsl}
gives the pathwise large deviations for the process
\begin{equation}\label{eq:HestonReps}
R^{\eps}=\frac{1}{\sqrt{\eps}h(\eps)}\left(\int_{0}^{\cdot}Y^{\eps}_{s}\D s-\theta\right).
\end{equation}
In this case, the Poisson equation~\eqref{Eq:uPoisson} satisfies $u_{y}(y)=-\frac{1}{\kappa}$, and hence
$\overline{Q}(x)\equiv \frac{\xi^{2}\theta}{\kappa^{2}}$,
yielding the rate function
\begin{equation}\label{eq:HestonIntRate}
  I(\phi) = \left\{
\begin{array}{ll}
\displaystyle  \frac{1}{2} \frac{\gamma^{2}\kappa^{2}}{\xi^{2}\theta}\int_0^T \left|\dot{\phi}_s\right|^{2} \D s,
  & \text{if }\phi \in \mathcal{AC}([0, T],\RR) \text{ and } \phi_{0}=0,\\
 +\infty, & \text{otherwise}.
 \end{array}
 \right.
\end{equation}
Take for simplicity $\gamma=1$, i.e. $\eps = \delta$, then we can write, for any $t\geq 0$,
$$
R_t^{\eps}
  = \frac{1}{\sqrt{\eps}h(\eps)}\left(\int_{0}^{t}Y^{\eps}_{s}\D s-\theta\right)
 = \frac{1}{\sqrt{\eps}h(\eps)}\int_{0}^{t}Y^{\eps}_{s}\D s - \theta_\eps
   = \frac{\sqrt{\eps}}{h(\eps)}\int_{0}^{t/\eps}Y_{u}\D s - \theta_\eps
   = \frac{\sqrt{\eps}}{h(\eps)}\Vv_{t/\eps} - \theta_\eps
$$
where we denote $\theta_\eps := \frac{\theta}{\sqrt{\eps}h(\eps)}$,
and $\Vv_{t} :=\int_{0}^{t}Y_u \D u$ represents the realised variance over the period $[0,t]$.
The large deviations for the sequence~$R^\eps$ therefore implies that, taking $t=1$,
and applying the contraction principle, for any fixed $x>0$,
$$
\lim_{\eps\downarrow 0} \frac{1}{h^2(\eps)}\log\PP\left(\frac{\sqrt{\eps}}{h(\eps)}\Vv_{1/\eps} - \theta_\eps \geq x \right) = -\inf\Big\{I(\phi): \phi_1\geq x\Big\},
$$
or, taking for example $h(\eps) = \eps^{-\beta}$ for $\beta\in (0,1/2)$,
and renaming $\eps\mapsto t^{-1}$,
$$
\lim_{t\uparrow \infty} t^{-2\beta}
\log\PP\left(\Vv_{t} \geq x t^{\beta+1/2} +\theta t \right) = -\inf\Big\{I(\phi): \phi_1\geq x\Big\}.
$$
Here again, this minimum can be computed in closed form from~\eqref{eq:HestonIntRate}.
The corresponding Euler-Lagrange equation yields the linear optimal path $\phi_t = x t$, for any $x> 0$, so that
$$
J_{\Vv}(x) := \inf\Big\{I(\phi): \phi_1\geq x\Big\} =
  \frac{\kappa^{2}x^2}{2\xi^{2}\theta} = \frac{x^2}{2\overline{Q}}.
$$
Note that, since the drift and diffusion of~$X$ in~\eqref{eq:Heston} do not depend on~$X$ itself,
then~$\overline{Q}$ is constant here.
In the Heston case~\eqref{eq:Heston}, we can compare this large-time MDP
to the corresponding large-time LDP for the realised variance.
Indeed, from~\cite[Proposition~6.3.4.1]{JCY}, the moment generating function for the realised variance reads,
for all~$u$ such that the right-hand side is finite,
$$
\Lambda(u,t):=\log\EE\left(\E^{u\Vv_{t}}\right) =
\frac{2\kappa\theta}{\xi^2}\log\left(\frac{2\gamma(u)\exp\left\{(\kappa+\gamma(u))\frac{t}{2}\right\}}{\gamma(u)\left(\E^{\gamma(u)t}+1\right)+\kappa\left(\E^{\gamma(u)t}-1\right)}\right)
+\frac{2u y_0\left(\E^{\gamma(u)t}-1\right)}{\gamma(u)\left(\E^{\gamma(u)t}+1\right)+\kappa\left(\E^{\gamma(u)t}-1\right)},
$$
where $\gamma(u):=\sqrt{\kappa^2-2\xi^2 u}$.
Straightforward computations yield that
$$
\Lambda_{\infty}(u) := \lim_{t\uparrow\infty}\frac{\Lambda(u,t)}{t}
 = \frac{\kappa\theta}{\xi^2}\left(\kappa-\gamma(u)\right),
\quad\text{for all }u<\frac{\kappa^2}{2\xi^2}.
$$
We can therefore use a partial version of the G\"artner-Ellis theorem~\cite[Section 2.3]{DZ98} to show that the sequence
$(t^{-1}\Vv_{t})_{t>0}$ satisfies a large deviations principle as~$t$ tends to infinity, with rate function~$\Lambda^*$ defined as the Fenchel-Legendre transform of~$\Lambda_{\infty}$.
More precisely, for any $x>0$,
$$
\Lambda^*(x) = \sup_{u<\frac{\kappa^2}{2\sigma^2}}\Big\{ux - \Lambda_{\infty}(u)\Big\}
 = \frac{\kappa^2(x-\theta)^2}{2\xi^2 x}
 \qquad\text{and}\qquad
\partial_{xx}\Lambda^*(x)\vert_{x=\theta} = \frac{\kappa^2}{\xi^2\theta} = \overline{Q}^{-1}.
$$

Again, we observe that the moderate deviations rate function characterises the local curvature of the large deviations rate function around its minimum.
We can translate this behaviour into asymptotics of options on the realised variance,
both in the large deviations case and in the moderate deviations one:
\begin{proposition}
Let $\beta \in (0,1/2)$. As~$t$ tends to infinity, we have, for all $x>0$,
$$
\lim_{t\uparrow\infty}\frac{1}{t}\log\EE\left(\E^{x t} - \E^{\Vv_t}\right)_+ = x - \Lambda^*(x)
\qquad\text{and}\qquad
\lim_{t\uparrow\infty}\left[\frac{1}{t^{2\beta}}\log\EE^{\PP}\left(\E^{x t^{\beta +1/2}} - \E^{\Vv_t-\theta t}\right)_+-\frac{x}{t^{\beta-1/2}}\right] = -J_{\Vv}(x).
 $$
\end{proposition}
\begin{remark}
Note that, again, formally setting $\beta=1/2$ in the second limit, which arises from moderate deviations,
one obtains after simplification,
$$
\lim_{t\uparrow\infty}\frac{1}{t}\log\EE\left(\E^{(x+\theta) t} - \E^{\Vv_t}\right)_+ = x+\theta - J_{\Vv}(x),
$$
which again, as in Remark~\ref{rem:MDPLDPCurv}, indicates that the moderate deviations rate function
represents exactly the curvature of the large deviations one at its minimum.
\end{remark}
\begin{proof}
We start with the first limit, which falls in the scope of large deviations.
Mimicking the proof of Proposition~\ref{prop:LargeTOption}, for all $t\geq 1$ and $\eps>0$, we can write the inequalities
$$
\E^{x t} \left(1-\E^{-\eps}\right)\ind_{\{\Vv_t/t < x-\eps\}}
\leq \left(\E^{x t} - \E^{\Vv_t}\right)_+
\leq \E^{x t} \ind_{\{\Vv_t/t < x\}}.
$$
Taking expectations, logarithms and dividing by~$t$ and taking limits, we obtain
$$
x + \lim_{t\uparrow\infty}\frac{1}{t}\log\PP\left(\frac{\Vv_t}{t} < x-\eps\right)
\leq \lim_{t\uparrow\infty}\frac{1}{t}\log\EE\left(\E^{x t} - \E^{\Vv_t}\right)_+
\leq x + \lim_{t\uparrow\infty}\frac{1}{t}\log\PP\left(\frac{\Vv_t}{t} < x\right).
$$
The large deviations obtained above for the sequence $(\Vv_t/t)_{t>0}$ as~$t$ tends to infinity concludes the proof.

We now prove the other limit, which arises from moderate deviations.
Now, mimicking the proof of Proposition~\ref{prop:LargeTOption}, for all $t\geq 1$ and $\eps>0$, we can write
$$
\E^{x t^{\beta +1/2}} \left(1-\E^{-\eps}\right)\ind_{\{\widetilde{\Vv}_t/t^{\beta +1/2} < x-\eps\}}
\leq \left(\E^{x t^{\beta +1/2}} - \E^{\widetilde{\Vv}_t}\right)_+
\leq \E^{x t^{\beta +1/2}} \ind_{\{\widetilde{\Vv}_t/t^{\beta +1/2} < x\}},
$$
with $\widetilde{\Vv}_t:=\Vv_t - \theta t$,
so that, taking expectations, logarithms and dividing by $t^{2\beta}$,
the proposition follows from
$$
\frac{x}{t^{\beta-1/2}} + \frac{1}{t^{2\beta}}\log\PP\left(\frac{\widetilde{\Vv}_t}{t^{\beta +1/2}} < x-\eps\right)
\leq \frac{1}{t^{2\beta}}\log\EE\left(\E^{x t^{\beta +1/2}} - \E^{\widetilde{\Vv}_t}\right)_+
\leq \frac{x}{t^{\beta-1/2}} + \frac{1}{t^{2\beta}}\log\PP\left(\frac{\widetilde{\Vv}_t}{t^{\beta +1/2}} < x\right).
$$
\end{proof}

\section{Control representations for the proof of Theorem~\ref{T:IntegratedFunctionsl}}\label{S:ControlRep}
In this section we connect Theorem~\ref{T:IntegratedFunctionsl} to certain stochastic control representation and to the Laplace principle.
We start by applying It\^{o} formula to the solution $u$ to~\eqref{Eq:uPoisson}.
Rearranging terms, we obtain
\begin{equation*}
\begin{array}{rl}
R_{t}^{\eps}
 & = \displaystyle \frac{\delta^{2}}{\eps^{3/2} h(\eps)}\left(u(X^{\eps}_{t},Y^{\eps}_{t})-u(x_{0},y_{0})\right)-\frac{\delta}{\eps h(\eps)}\int_{0}^{t}(g u_{y})(X^{\eps}_{s},Y^{\eps}_{s})\D Z_s
- \frac{\delta^{2}}{\eps h(\eps)}\int_{0}^{t}(\sigma u_{x})(X^{\eps}_{s},Y^{\eps}_{s})\D W_s\\
 & \displaystyle - \frac{\rho\delta}{\sqrt{\eps} h(\eps)}\int_{0}^{t}(\sigma g u_{xy}) (X^{\eps}_{s},Y^{\eps}_{s})\D s
 + \frac{\delta}{2\sqrt{\eps}{h(\eps)}}\int_{0}^{t}\left(\sigma^{2}u_{x}\right)(X^{\eps}_{s},Y^{\eps}_{s})\D s
 - \frac{\delta^{2}}{2\sqrt{\eps}h(\eps)}\int_{0}^{t}\left(\sigma^{2}u_{xx}\right)(X^{\eps}_{s},Y^{\eps}_{s})\D s.
\end{array}
\end{equation*}
We shall consider $R_{t}^{\eps}$ through this expression together with~\eqref{eq:Main3c} as the triple $(R^{\eps}_{t},X^{\eps}_{t},Y^{\eps}_{t})$.
By~\cite[Section 1.2]{DE97}, the large deviations principle (with speed~$h^2(\eps)$) for $R^\eps$ is equivalent to the Laplace principle, which states that for any bounded continuous function
$G$ mapping $\mathcal{C}([0, T];\RR)$ to~$\RR$,
\begin{equation} \label{E:LaplacePrinciple}
  \lim_{\eps \downarrow 0} - \frac{1}{h^2(\eps)} \log \mathbb{E} \left[ \exp\left\{- h^2(\eps) G(R^\eps) \right\} \right]
   = \inf_{\phi \in \mathcal{C}([0, T]; \RR)} \Big\{ I(\phi) + G(\phi) \Big\},
\end{equation}
where $I(\cdot)$ is called the action functional.
We prove~\eqref{E:LaplacePrinciple} and then Theorem~\ref{T:IntegratedFunctionsl} identifies~$I(\phi)$.
The proof of~\eqref{E:LaplacePrinciple} is based on appropriate stochastic control representations developed in~\cite{BD98}.
Let~$\mathcal{A}$ be the space of $\mathscr{F}_t$-progressively measurable two-dimensional processes
$v = (v_1, v_2)$ such that
$\EE\left(\int_0^T \lvert v(s) \rvert^2 \D s\right)$ is finite.
From~\cite[Theorem 3.1]{BD98}, we can write the stochastic control representation
\begin{align}
- \frac{1}{h^2(\eps)} \log \mathbb{E} \left[ \exp\left\{- h^2(\eps) G(R^\eps) \right\} \right] = \inf_{u^\eps \in \mathcal{A}} \mathbb{E} \left[ \frac{1}{2} \int_0^T \lvert u^\eps(s) \rvert^2 \D s + G(R^{\eps,u^\eps}) \right]\label{Eq:ControlRepresentation}
\end{align}
where the controlled process $(R^{\eps,u^\eps}_{t},X^{\eps,u^\eps}_{t},Y^{\eps,u^\eps}_{t})$
satisfies the system
\begin{align}\label{Eq:ControlledProcesses}
R_{t}^{\eps,u^\eps}
 & = \displaystyle
 \frac{\delta^{2}}{\eps^{3/2} h(\eps)}\left(u(X^{\eps,u^\eps}_{t},Y^{\eps,u^\eps}_{t})-u(x_{0},y_{0})\right)-\frac{\delta}{\eps h(\eps)} \int_{0}^{t}(g u_{y})(X^{\eps,u^\eps}_{s},Y^{\eps,u^\eps}_{s})\D Z_s\nonumber\\
 & \displaystyle
-\frac{\delta}{\eps}\rho\int_{0}^{t}(g u_{y})(X^{\eps,u^\eps}_{s},Y^{\eps,u^\eps}_{s})u^{\eps}_{1}(s)\D s-\frac{\delta}{\eps}\rrho\int_{0}^{t}(g u_{y})(X^{\eps,u^\eps}_{s},Y^{\eps,u^\eps}_{s})u^{\eps}_{2}(s)\D s\nonumber\\
 & \displaystyle
 -\frac{\delta^{2}}{\eps}\rho\int_{0}^{t}(\sigma u_{x})(X^{\eps,u^\eps}_{s},Y^{\eps,u^\eps}_{s})u^{\eps}_{1}(s)\D s\nonumber\\
 & \displaystyle
 -\rho\frac{\delta}{\sqrt{\eps} h(\eps)}\int_{0}^{t}(\sigma g u_{xy}) (X^{\eps,u^\eps}_{s},Y^{\eps,u^\eps}_{s})\D s+\frac{1}{2}\frac{\delta}{\sqrt{\eps}{h(\eps)}}\int_{0}^{t}\left(\sigma^{2}u_{x}\right)(X^{\eps,u^\eps}_{s},Y^{\eps,u^\eps}_{s})\D s
\\
 & \displaystyle
 -\frac{1}{2}\frac{\delta^{2}}{\sqrt{\eps}h(\eps)}\int_{0}^{t}\left(\sigma^{2}u_{xx}\right)(X^{\eps,u^\eps}_{s},Y^{\eps,u^\eps}_{s})\D s-\frac{\delta^{2}}{\eps h(\eps)}\int_{0}^{t}(\sigma u_{x})(X^{\eps,u^\eps}_{s},Y^{\eps,u^\eps}_{s})\D W_s,\nonumber\\
\D X^{\eps,u^\eps}_t
 & = \displaystyle
 -\frac{1}{2}\frac{\eps}{\delta}\sigma^2(X^{\eps,u^\eps}_t,Y^{\eps,u^\eps}_t) \D t + \sqrt{\eps}h(\eps)\sigma(X^{\eps,u^\eps}_t,Y^{\eps,u^\eps}_t)u_{1}^{\eps}(t)\D t+ \sqrt{\eps}\sigma(X^{\eps,u^\eps}_t,Y^{\eps,u^\eps}_t)\D W_t,\nonumber\\
\D Y^{\eps,u^\eps}_t
 & = \displaystyle
 \frac{\eps}{\delta^{2}}f(X^{\eps,u^\eps}_t,Y^{\eps,u^\eps}_t) \D t
 + \frac{\sqrt{\eps}h(\eps)}{\delta}g(X^{\eps,u^\eps}_t,Y^{\eps,u^\eps}_t)
 \left[\rho u_{1}^{\eps}(t) + \rrho u_{2}^{\eps}(t)\right] \D t
 + \frac{\sqrt{\eps}}{\delta}g(X^{\eps,u^\eps}_t,Y^{\eps,u^\eps}_t)\D Z_t,\nonumber
\end{align}
with $\D\langle W, Z\rangle_t = \rho \D t$.
With this control representation at hand, we proceed by analysing
the limit of the right-hand side of~\eqref{Eq:ControlRepresentation} as~$\eps$ tends to zero.
Similarly to~\cite{MS}, we define the function
$\theta: \RR^{4}\mapsto \RR$ by
\begin{align}
\theta(x,y,z_{1},z_{2}) := -\gamma^{-1} g(x,y)u_{y}(x,y)\left(\rho z_{1} + \rrho z_{2}\right).
\label{Eq:GoverningLimitODE}
\end{align}
Let $\Zz=\RR$ define the control space,
and $\Yy$ the state space of~$Y$.
Our assumptions guarantee that~$\theta$ is bounded in $x$, affine in $z_{1},z_{2}$ growing at most polynomially in $y$. Next step is to introduce an appropriate family of occupation measures whose role is to single out the correct averaging taking place in the limit. For this reason, let $0<\Delta=\Delta(\eps)\downarrow 0$  be a parameter whose role is to exploit a time-scale separation.
Let~$A_{1}$, $A_{2}$, $B$, $\Gamma $ be Borel sets of $\Zz$, $\Zz$, $\Yy$, $[0,T]$ respectively.
Then, define the occupation measure  $P^{\eps,\Delta}$ by
\begin{equation} \label{E:occupation}
  P^{\eps, \Delta}(A_1 \times A_2 \times B \times \Gamma)
   := \int_\Gamma \left[ \frac{1}{\Delta} \int_t^{t+\Delta} \ind_{A_1}(u_1^\eps(s)) \ind_{A_2}(u_2^\eps(s)) \ind_B(Y_s^{\eps, u^\eps}) \D s \right] \D t,
\end{equation}
and assume $u_i^\eps(s) = 0$ if $s > T$. For a Polish space $S$, let $\mathcal{P}(S)$ be the space of probability measures on $S$. Next we recall the definition of a viable pair for the moderate deviations case.

\begin{definition} [Definition 4.1 in~\cite{MS}]\label{D:ViablePair}
Let $\theta(x, y, z_1, z_2) \colon \RR \times \Yy \times  \Zz \times \Zz \to \RR$ be a function that grows at most linearly in~$y$.
For each $x\in\RR$, let $\Ll_{x}$ be a second-order elliptic partial differential operator and denote by $\mathcal{D}(\Ll_{x})$ its domain of definition. A pair $(\psi, P) \in \mathcal{C}([0,T]; \RR) \times \mathcal{P}(\Zz \times \Zz \times \Yy \times [0,T])$ is called a viable pair with respect to $(\theta, \Ll_{x})$,
and we write $(\psi, P)\in\mathcal{V}(\theta, \Ll_{x})$, if
\begin{itemize}
\item the function $\psi$ is absolutely continuous;
\item the measure $P$ is integrable in the sense that
\begin{equation*}
  \int_{\Zz \times \Zz \times \Yy \times [0,T]}
  \left( \lvert z_1 \rvert^2 + \lvert z_2 \rvert^2 + \lvert y \rvert^{2}\right) P(\D z_1 \,\D z_2 \,\D y \,\D s) < \infty;
\end{equation*}
\item for all $t \in [0,T]$ (with~$\overline{X}$ defined in~\eqref{E:averagedSDE1}),
\begin{equation}\label{E:viableSDE}
  \psi_t = \int_{\Zz \times \Zz \times \Yy \times [0,t]} \theta(\overline{X}_s, y, z_1, z_2) P(\D z_1 \,\D z_2 \,\D y \,\D s);
\end{equation}
\item for all $t \in [0, T]$ and for every $F \in \mathcal{D}(\Ll_{x})$,
\begin{equation}\label{E:viableCenter}
  \int_0^t \int_{\Zz \times \Zz \times \Yy} \Ll_{\overline{X}_s} F(y) P(\D z_1 \,\D z_2 \,\D y \,\D s) = 0;
\end{equation}
\item for all $t \in [0, T]$,
\begin{equation}\label{E:viableLebesgue}
  P(\Zz \times \Zz \times \mathcal{Y} \times [0,t]) = t.
\end{equation}
\end{itemize}
\end{definition}
The last item is equivalent to stating that the last marginal of $P$ is Lebesgue measure, or that~$P$ can be decomposed as $P(\D z_1 \,\D z_2 \,\D y \,\D t) = P_t(\D z_1\, \D z_2\, \D y) \,\D t$.  Then we shall establish the following result:
\begin{theorem} \label{T:control}
Under Assumptions~\ref{C:ergodic1},~\ref{C:StrongSolution1}, \ref{C:GrowthAssumptions} and~\ref{C:Positiveness}, the family of processes $\{R^{\eps},\eps>0\}$ from~\eqref{Eq:IntergatedProcess} satisfies the large deviations principle, with action functional
\begin{equation}
  I(\phi) = \inf_{(\phi,P) \in \mathcal{V}(\theta, \Ll_{x})}
  \left\{ \frac{1}{2} \int_{\Zz \times \Zz \times \Yy \times [0, T]}
  \left( \lvert z_1 \rvert^2 + \lvert z_2 \rvert^2 \right) P(\D z_1 \, \D z_2 \, \D y \, \D s) \right\}\label{Def:RateFunctionControlForm}
\end{equation}
with the convention that the infimum over the empty set is infinite.
\end{theorem}

As will be shown during the proof, Theorem~\ref{T:IntegratedFunctionsl} follows directly from Theorem~\ref{T:control}.

\section{Proof of  Theorem~\ref{T:IntegratedFunctionsl}}\label{S:ProofIntegratedFunctMDP}
In this section we offer the proof of Theorem ~\ref{T:IntegratedFunctionsl}.
As mentioned in Section~\ref{S:ControlRep}, the proof of Theorem~\ref{T:IntegratedFunctionsl}
is equivalent to that of the Laplace principle in Theorem~\ref{T:control}.
In Subsections~\ref{SS:Tightness} and~\ref{SS:ExistenceViablePair} we prove tightness and convergence of the pair $(R^{\eps, u^\eps}, P^{\eps, \Delta})$ respectively. In Subsection~\ref{SS:LaplacePrinciple}, we finally establish the Laplace principle and the representation formula of Theorem~\ref{T:IntegratedFunctionsl}.
The proofs of these results make use of the results of~\cite{MS}. Below we present the main arguments, highlighting the differences and give exact pointers to~\cite{MS} when appropriate.

\subsection{Tightness of the pair $\{(R^{\eps, u^\eps}, P^{\eps, \Delta}),\eps>0\}$}\label{SS:Tightness}
In this section we prove the following proposition:
\begin{proposition}\label{P:tightness}
Under Assumptions~\ref{C:ergodic1},~\ref{C:StrongSolution1} and~\ref{C:GrowthAssumptions},
for a family $\{u^{\eps},\eps>0\}$ of controls in~$\mathcal{A}$ satisfying
\begin{equation*}
\sup_{\eps>0}\int_{0}^{T}\left| u^{\eps}(t)\right|
^{2}\D t<N \quad\text{almost surely},
\end{equation*}
for some $N<\infty$, the following hold:
\begin{enumerate}[(i)]
\item {the family $\{(R^{\eps,u^{\eps}},\mathrm{P}^{\eps,\Delta
}),\eps>0\}$ is tight on $\mathcal{C}([0,T];\RR)\times \mathcal{P}(\Zz \times \Zz \times \Yy \times [0,T])$;}

\item {With the set
\[
\mathcal{B}_{M} := \Big\{(z_{1},z_{2},y)\in\Zz\times\Zz\times \RR: |z_{1}|>M, |z_{2}| > M, |y| >M\Big\},
\]
the family $\{\mathrm{P}^{\eps,\Delta},\eps>0\}$ is uniformly
integrable in the sense that
\[
\lim_{M\uparrow\infty}\sup_{\eps>0}\mathbb{E}_{x_{0},y_{0}}\left[
\int_{(z_{1},z_{2},y)\in\mathcal{B}_{M}\times [0,T]}
\left[\left| z_{1}\right| +\left| z_{2}\right|+\left| y\right| \right]
\mathrm{P}^{\eps,\Delta}(\D z_{1} \,\D z_{2} \,\D y \,\D t)\right]  =0.
\]
}
\end{enumerate}
\end{proposition}

\begin{proof}[Proof of Proposition~\ref{P:tightness}]
Tightness of $\{ P^{\eps, \Delta}, \eps,\Delta>0 \}$ on $\mathcal{P}(\Zz \times \Zz \times \Yy \times [0,T])$
and the uniform integrability of the family of occupation measures
is the subject of~\cite[Section 5.1.1]{MS}, and will not be repeated here.
It remains to prove tightness of $\{ R^{\eps, u^\eps},\eps>0 \}$ on $\mathcal{C}([0,T];\RR)$.
We prove it making use of the characterisation of~\cite[Theorem 8.7]{Billingsley1968},
and it follows if we establish that there is $\eps_{0}>0$ such that for every $\eta>0$,
\begin{enumerate}[(i)]
\item there exists $N<\infty$ such that
\begin{equation}
\mathbb{P}\left(\sup_{0\leq t\leq T}\left|R^{\eps,u^{\eps}}_{t}\right|>N\right) \leq \eta,
\quad\textrm{ for every }\eps\in(0,\eps_{0});
\label{Eq:CompactContainment}
\end{equation}
\item for every $M<\infty$,
\begin{equation}
\lim_{\alpha\downarrow0}\sup_{\eps\in(0,\eps_{0})}\PP\left(
\sup_{|t_{1}-t_{2}|<\alpha,0\leq t_{1}<t_{2}\leq T}\left|R^{\eps,u^{\eps}}_{t_{1}}-R^{\eps,u^{\eps}}_{t_{2}}\right|\geq\eta, \sup_{t\in[0,T]}\left|R^{\eps,u^{\eps}}_{t}\right|\leq M\right)  =0.\label{Eq:ContinuousReg}
\end{equation}
\end{enumerate}
Essentially both statements follow from the control representation~\eqref{Eq:ControlledProcesses} together with the results on the growth of the solution to the Poisson equation by Theorem~\ref{T:regularity} and Lemma~\ref{T:regularityCIR} and using Lemma~\ref{L:productBound} to treat each term on the right-hand side of~\eqref{Eq:ControlledProcesses}. For sake of completeness, we prove the first statement~\eqref{Eq:CompactContainment}.
The second statement~\eqref{Eq:ContinuousReg} follows similarly using the general purpose
Lemma~\ref{L:productBound}. Rewrite~\eqref{Eq:ControlledProcesses} as
\begin{equation}\label{Eq:ControlRepresentation_R}
R_{t}^{\eps,u^\eps} = \sum_{i=1}^{9}R_{t}^{i,\eps},
\end{equation}
where $R_{t}^{i,\eps}$ represents the $i^{\text{th}}$ term
on the right-hand side of~\eqref{Eq:ControlledProcesses}, and notice that
\[
\mathbb{P}\left(\sup_{0\leq t\leq T}\left|R^{\eps,u^{\eps}}_{t}\right|>N\right)\leq \sum_{i=1}^{9}\mathbb{P}\left(\sup_{0\leq t\leq T}\left|R^{i,\eps}_{t}\right|>\frac{N}{9}\right).
\]
Theorem~\ref{T:regularity} and an application of H\"{o}lder inequality imply that, for $q_{H}<2-q_{g}$,
\begin{align}
\mathbb{E}\left(\sup_{t\in[0,T]}| R^{1,\eps}_{t}|\right)
 & \leq \frac{2\delta^{2}}{\eps^{3/2} h(\eps)}\left(1+\mathbb{E}\sup_{t\in[0,T]}| Y^{\eps,u^\eps}_{t}|^{q_{H}}\right)\leq
 \frac{2\delta^{2}}{\eps^{3/2} h(\eps)}\left(1+\left(\mathbb{E}\sup_{t\in[0,T]}| Y^{\eps,u^\eps}_{t}|^{2}\right)^{q_{H}/2}\right)\nonumber\\
 &\leq K \frac{\delta^{2}}{\eps^{3/2} h(\eps)}\left(1+\left(\frac{\sqrt{\eps}}{\delta}\right)^{(1-\nu)q_{H}}+h(\eps)^{\frac{q_{H}}{1-q_{g}}}\right)\nonumber\\
 & = K \left(\frac{\delta^{2}}{\eps^{3/2} h(\eps)}+ \left(\frac{\delta}{\eps}\right)^{2-(1-\nu)q_{H}}\frac{\eps^{\frac{1-(1-\nu)q_{H}}{2}}}{ h(\eps)}+\frac{\delta^{2}}{\eps^{2}} \sqrt{\eps}h(\eps)^{\frac{q_{g}+q_{H}-1}{1-q_{g}}}\right) ,\label{Eq:FirstTermConv}
\end{align}
where the third inequality follows from Lemma~\ref{L:Ygrowth}.
This certainly stays bounded (actually, it goes to zero) by Assumption~\ref{C:GrowthAssumptions} and because, for any $q_{H}<2-q_{g}$ we can choose $\nu\in(0,1-q_{g})$ so that $\lim_{\eps\downarrow  0}\frac{\eps^{\frac{1-(1-\nu)q_{H}}{2}}}{ h(\eps)}=0$.
Hence we obtain that  for some unimportant constant $K<\infty$,
$$
 \PP\left( \sup_{t\in[0,T]} \left|R^{1,\eps}_{t}\right| >  \frac{N}{9} \right) \leq \frac{K}{N}.
$$
For the stochastic integral term $R^{2,\eps}_{t}$, we have that for a constant $C<\infty$ that may change from line to line and  for $\nu>0$ small enough such that $q_{gu_{y}}(1+\nu)<1$ (with some abuse of notation $q_{gu_{y}}$ is the degree of polynomial growth in $|y|$ of the function $g(x,\cdot)u_{y}(x,\cdot)$), we have
\begin{align*}
\PP &\left( \sup_{t\in[0,T]} \left\lvert \int_{0}^{t} gu_{y}(X_{s}^{\eps, u^\eps}, Y_{s}^{\eps, u^\eps}) \D Z_{s} \right\rvert > \frac{\eps h(\eps)N}{9\delta}\right)\leq\nonumber\\
  & \leq \frac{C}{N^{2(1+\nu)}}\left(\frac{\delta}{\eps h(\eps)}\right)^{2(1+\nu)}
  \EE\left(\sup_{t\in[0,T]} \left\lvert \int_{0}^{t} gu_{y}(X_{s}^{\eps, u^\eps}, Y_{s}^{\eps, u^\eps}) \D W_{s} \right\rvert^{2(1+\nu)}  \right)\nonumber\\
  &\leq \frac{C}{N^{2(1+\nu)}}\left(\frac{\delta}{\eps h(\eps)}\right)^{2(1+\nu)}
    \EE\left(\sup_{t\in[0,T]} \left\lvert \int_{0}^{t} \left|gu_{y}(X_{s}^{\eps, u^\eps}, Y_{s}^{\eps, u^\eps})\right|^{2} \D s \right\rvert^{(1+\nu)}  \right),
\end{align*}
from which the result follows by Lemma~\ref{L:productBound} given that $q_{gu_{y}}<1$.
Similarly, we obtain a corresponding bound for the other stochastic integral term $R^{9,\eps}_{t}$.
The Riemann integral terms $R^{3,\eps}_{t},\ldots, R^{8,\eps}_{t}$ are treated very similarly again using Lemma~\ref{L:productBound}. The conditions on $q_{\sigma},q_{g}$ and $q_{H}$
from Assumption~\ref{C:GrowthAssumptions} guarantee that Lemma~\ref{L:productBound} applies in these cases.  These considerations yield~\eqref{Eq:CompactContainment}.

The second statement~\eqref{Eq:ContinuousReg} follows by similar arguments using again Lemma~\ref{L:productBound} and the growth properties of the involved functions with respect to~$|y|$.
The only term that potentially needs some discussion is the~$R^{1,\eps}_{t}$ term.
In particular, we need to show that for every $\eta>0$, there exists $\eps_{0}>0$ such that
$$
\sup_{\eps\in(0,\eps_{0})}\PP\left( \sup_{t\in[0,T]} \left|R^{1,\eps}_{t}\right| >  \eta \right) \leq \eta.
$$
But this follows again by the estimate on $\mathbb{E}\left(\sup_{t\in[0,T]}| R^{1,\eps}_{t}|\right)$ above together with Assumption~\ref{C:GrowthAssumptions}. This completes the proof of the proposition.
\end{proof}

\subsection{Convergence of the pair $\{(R^{\eps, u^\eps}, P^{\eps, \Delta}),\eps>0\}$}\label{SS:ExistenceViablePair}

In Section~\ref{SS:Tightness} we proved that the family of processes $\{ (R^{\eps, u^\eps}, P^{\eps, \Delta}),\ \eps > 0 \}$ is tight.
It follows that for any subsequence of $\eps$ converging to 0, there exists a subsubsequence
of $(R^{\eps, u^\eps}, P^{\eps, \Delta})$ which converges in distribution to some limit
$(\overline{R}, \overline{P})$.
The goal of this section is to show that $(\overline{R}, \overline{P})$
is a viable pair with respect to $(\theta,\Ll_{x})$, according to Definition~\ref{D:ViablePair}.
To show that the limit point $(\overline{R}, \overline{P})$ is a viable pair,
we must show that it satisfies~\eqref{E:viableSDE},~\eqref{E:viableCenter}, and~\eqref{E:viableLebesgue}.
The proof of~\eqref{E:viableCenter} and~\eqref{E:viableLebesgue}
can be found in~\cite[Section 5.2]{MS}, whereas the proof of statement~\eqref{E:viableSDE}
follows via the Skorokhod Representation Theorem~\cite[Theorem 6.7]{Billingsley1968}
and the martingale problem formulation.
For completeness, we present the argument below.
We will invoke the Skorokhod Representation Theorem, which allows us to assume that there exists a probability space in which the desired convergence occurs with probability one. Let us define
$$
\Psi^{\eps, u^\eps}_{t} := R^{\eps, u^\eps}_{t}-\frac{\delta^{2}}{\eps^{3/2} h(\eps)}\left(u(X^{\eps,u^\eps}_{t},Y^{\eps,u^\eps}_{t})-u(x_{0},y_{0})\right).
$$
As in the proof of (\ref{Eq:FirstTermConv}),
$$
\lim_{\eps\downarrow 0}\frac{\delta^{2}}{\eps^{3/2} h(\eps)}
\EE\left(\sup_{t\in[0,T]}\left|u(X^{\eps,u^\eps}_{t},Y^{\eps,u^\eps}_{t})-u(x_{0},y_{0})\right|\right)=0,
$$
so that it is enough to consider the limit in distribution of the family
$\{ \Psi^{\eps, u^\eps}_{\cdot},\eps > 0 \}$. The rest of the argument is now classical.

Consider an arbitrary function $G$  that is real valued, smooth with compact support on~$\RR$.
Fix two positive integers~$p_1$ and~$p_2$ and let $\phi_j$, $j = 1, \dots, p_1$, be real valued, smooth functions with compact support on $\Zz \times \Zz \times \Yy \times [0, T]$.
Let~$S_{1}$, $S_{2}$, and $t_i$, $i = 1, \dots, p_2$, be nonnegative real numbers with
$t_i \le S_{1} <  S_{1} + S_{2} \le T$.
Let~$\zeta$ be a real-valued, bounded and continuous function with compact support on $\RR^{p_2} \times \RR^{p_1 p_2}$.
Given a measure $P_{0} \in \mathcal{P}(\Zz \times \Zz \times \Yy \times [0,T])$ and $t \in [0,T]$, define
$$
\left(P_{0}, \phi_j\right)_t := \int_{\Zz \times \Zz \times \Yy \times [0,t]} \phi_j(z_1, z_2, y, s)P_{0}(\D z_1 \,\D z_2 \,\D y \,\D s),
$$
the empirical measure
$$
P_t^{\eps, \Delta}(\D z_1 \,\D z_2 \,\D y) := \frac{1}{\Delta} \int_t^{t+\Delta} \ind_{\D z_1}(u_1^\eps(s)) \ind_{\D z_2}(u_2^\eps(s)) \ind_{\D y}(Y_s^{\eps, u^\eps}) \D s,
$$
as well as the operator $\widehat{\Ll}_t^{\eps, \Delta}$ as
$$
\widehat{\Ll}_t^{\eps, \Delta}G(\Psi) := \int_{\Zz \times \Zz \times \Yy}   G^{\prime}(\Psi)  \theta (\overline{X}_t, y, z_1, z_2) P_t^{\eps, \Delta}(\D z_1 \,\D z_2 \,\D y).
$$
Based on the martingale problem formulation, proving~\eqref{E:viableSDE} is equivalent to proving that
\begin{equation}\label{E:martingaleProblem}
 \lim_{\eps \downarrow 0} \EE\left[
 \zeta\left(\Psi_{t_i}^{\eps, u^\eps}, \left((P^{\eps, \Delta}, \phi_j)_{t_i}\right)_{i \le p_2, j \le p_1}\right)
 \left[ G(\Psi_{S_{1} + S_{2}}^{\eps, u^\eps}) - G({\Psi_{S_{1}}^{\eps, u^\eps}}) - \int_{S_{1}}^{S_{1} + S_{2}} \widehat{\Ll}_t^{\eps, \Delta}G(\Psi_{t}^{\eps, u^\eps}) \D t \right] \right] = 0,
\end{equation}
and
\begin{equation}\label{E:PConvergence}
 \lim_{\eps \downarrow 0}\left\{\int_{S_{1}}^{S_{1} + S_{2}} \widehat{\Ll}_t^{\eps, \Delta}G(\Psi_{t}^{\eps, u^\eps}) \D t - \int_{\Zz \times \Zz \times \RR \times [S_{1}, S_{1} + S_{2}]}   G^{\prime}( \overline{R}_t) \theta(\overline{X}_t, y, z_1, z_2) \overline{P}(\D z_1 \,\D z_2 \,\D y \,\D t) \right\} = 0.
\end{equation}
Note first that for every real-valued, continuous function~$\phi$ with compact support and $t \in [0, T]$,
$\left(P^{\eps, \Delta}, \phi\right)_t$ converges to $\left(\overline{P}, \phi\right)_t$
with probability one as~$\eps$ tends to zero.
Now~\eqref{E:PConvergence} follows directly by the first statement of~\cite[Lemma 5.1]{MS}.
In order to prove~\eqref{E:martingaleProblem} we first apply the It\^{o} formula to $G(\Psi)$.
Then it is easy to see that all of the terms apart from
$$
\frac{\delta}{\eps}\int_{0}^{t}
G^{\prime}( \Psi_s)(g u_{y})(X^{\eps,u^\eps}_{s},Y^{\eps,u^\eps}_{s})\left(\rho u^{\eps}_{1}(s)
 + \rrho u^{\eps}_{2}(s)\right)\D s
$$
in the representation~\eqref{Eq:ControlRepresentation_R} will vanish in the limit.
This term surviving in the limit together with the second statement of~\cite[Lemma 5.1]{MS}
directly yield~\eqref{E:martingaleProblem}. This concludes the proof of~\eqref{E:viableSDE}.

\subsection{Laplace principle and compactness of level sets}\label{SS:LaplacePrinciple}
We prove here the Laplace principle lower and upper bounds and the compactness of level sets of the action functional~$I(\cdot)$.
We start with the lower bound, i.e. we need to show that for all bounded, continuous functions~$G$ mapping $\mathcal{C}([0,T]; \RR)$ into~$\RR$,
$$
\liminf_{\eps \downarrow 0} - \frac{\log \EE \left[ \exp \left\{ - h^2(\eps) G(R^\eps) \right\} \right]}{h^2(\eps)}
 \ge \inf_{(\phi, P) \in \mathcal{V}(\theta, \Ll_{x})}
 \left\{\frac{1}{2} \int \left[ \lvert z_1\rvert^2 + \lvert z_2\rvert^2 \right] P(\D z_1\, \D z_2\, \D y \,\D s) + G(\phi) \right\}.
$$
It is sufficient to prove it along any subsequence such that
$- \frac{1}{h^2(\eps)} \log \EE \left[\E^{ - h^2(\eps) G(R^\eps) } \right]$
converges, which exists due to the uniform bound on the test function $G$. In addition, by Lemma~\ref{L:uBound}, we may assume that
\begin{equation*}
  \sup_{\eps > 0} \EE\int_0^1 \lvert u^\eps(s) \rvert^2 \D s \le N,
\end{equation*}
for some constant $N$. For such controls, we construct the family of occupation measures $P^{\eps, \Delta}$
from~\eqref{E:occupation}, and per Proposition~\ref{P:tightness} the family $\{(R^{\eps, u^\eps}, P^{\eps, \Delta}), \eps >0 \}$ is tight.
This indicates that for any subsequence, there is a further subsequence for which
$  (R^{\eps, u^\eps}, P^{\eps, \Delta})$ converges to $(\overline{R}, \overline{P})$
in distribution, with $(\overline{R}, \overline{P}) \in \mathcal{V}(\theta, \Ll_{x})$ being a viable pair
per Definition~\ref{D:ViablePair}. Then using Fatou's lemma, we conclude the proof of the lower bound
\begin{align*}
\liminf_{\eps \downarrow 0} - \frac{\log \EE \left[ \exp \left\{ - h^2(\eps) G(R^\eps) \right\} \right]}{h^2(\eps)}
    &\ge \liminf_{\eps \downarrow 0} \left\{ \EE\left( \frac{1}{2} \int_0^T \lvert u^\eps(s)\rvert^2 \,\D s + G(R^{\eps,u^\eps}) \right) - \eps \right\} \notag \\
    &\ge \liminf_{\eps \downarrow 0}\EE\left( \frac{1}{2} \int_0^T \frac{1}{\Delta} \int_t^{t+\Delta} \lvert u^\eps(s)\rvert^2 \D s \D t + G(R^{\eps,u^\eps}) \right)\\
    &= \liminf_{\eps \downarrow 0}\EE\left( \frac{1}{2} \int_{\Zz^2 \times \Yy \times [0,T]} \left[ \lvert z_1\rvert^2 +  \lvert z_2\rvert^2 \right] P^{\eps, \Delta} (\D z_1\,\D z_2\,\D y\,\D t) + G(R^{\eps,u^\eps}) \right)\\
    &\ge \EE\left( \frac{1}{2} \int_{\Zz^2\times \Yy \times [0,T]} \left[ \lvert z_1\rvert^2 +  \lvert z_2\rvert^2 \right] \overline{P} (\D z_1\,\D z_2\,\D y\,\D t) + G(\overline{R}) \right) \notag \\
    &\ge \inf_{(\xi, P) \in \mathcal{V}(\theta, \Ll_{x})} \left\{ \frac{1}{2} \int_{\Zz^2\times \Yy \times [0,T]} \left[ \lvert z_1\rvert^2 +  \lvert z_2\rvert^2 \right] P(\D z_1\,\D z_2\,\D y\,\D t) + G(\phi) \right\}. \notag
\end{align*}
Next we want to prove compactness of level sets. Namely, we want to show that for each $s < \infty$, the set
$\Phi_s := \{ \xi \in \mathcal{C}([0,T]; \RR) : I(\phi) \le s \}$
is a compact subset of $\mathcal{C}([0,T]; \RR)$,
which amounts to showing that it is  precompact and closed.
The proof is analogous to the proof of the lower bound using Fatou's lemma on the form of~$I(\cdot)$
as given by~\eqref{Def:RateFunctionControlForm} and thus the details are omitted; see also~\cite{DS12} for further details.

It remains to show the upper bound. From~\cite{MS} we can write
$I(\phi) = \int_0^T L^{o}(\overline{X}_s, \phi_s, \dot{\phi}_s) \D s$, where
\begin{equation*}
\begin{array}{rl}
  L^o(x, \eta, \beta) &:= \displaystyle \inf_{(v, \mu) \in \mathcal{A}_{x,\eta,\beta}^o} \frac{1}{2} \int_\mathcal{Y} \lvert v(y) \rvert^2 \mu_{x}(\D y),\\
  \mathcal{A}_{x, \eta,\beta}^o &:=
  \displaystyle \left\{
v = (v_1, v_2): \Yy\to\RR^{2},
  \mu \in \mathcal{P}(\Yy), \text{ such that}
  \begin{array}{rl}
  & \displaystyle \int_{\Yy} \Ll^{Y}_{x} F(y) \mu_{x}(\D y) = 0, \text{ for all } F \in \Cc^2(\Yy),\\
  & \displaystyle \int_{\Yy} \left[\lvert v(y)\rvert^2 +|y|^{2}\right]\mu_{x}(\D y) < \infty, \\
  & \displaystyle \beta = \int_{\Yy} \theta(x, y, v_1(y), v_2(y)) \mu_{x}(\D y).
  \end{array}
  \right\}
  \end{array}
\end{equation*}
Then, applying~\cite[Theorem 5.1]{MS} we get that if $\overline{Q}(x)$ defined in~\eqref{Eq:BarQDef}  is non-degenerate everywhere then
\begin{equation}
  L^o(x, \eta, \beta) = \frac{\beta^{2}}{2\overline{Q}(x)}, \label{Eq:EquivalenceInRep}
\end{equation}
and the infimum is attained at
$$
  u_1(y) =  -\frac{\rho\beta g(x,y)u_{y}(x,y)}{\gamma\overline{Q}(x)}
  \qquad\text{and}\qquad
  u_2(y) = -\frac{\rrho\beta g(x,y)u_{y}(x,y)}{\gamma\overline{Q}(x)}.
$$
This representation essentially gives the equivalence between Theorems~\ref{T:IntegratedFunctionsl} and~\ref{T:control}. Now, we have all the tools needed to prove the Laplace principle upper bound.
We need to show that for all bounded, continuous functions~$G$ mapping $\mathcal{C}([0,T]; \RR)$ into~$\RR$,
\begin{equation*}
  \limsup_{\eps \downarrow 0} - \frac{1}{h^2(\eps)} \log \mathbb{E} \left[ \exp \left\{ -h^2(\eps) G(R^\eps) \right\} \right] \le \inf_{\phi \in \mathcal{C}([0,T]; \RR)} [I(\phi) + G(\phi)].
\end{equation*}
Let $\zeta > 0$ be given and consider $\psi \in \mathcal{C}([0,T];\RR)$ with $\psi_0 = 0$ such that
\begin{equation*}
  I(\psi) + G(\psi) \le \inf_{\phi \in \mathcal{C}([0,T]: \RR)} [I(\phi) + G(\phi)] + \zeta < \infty.
\end{equation*}
Since~$G$ is bounded, this implies that~$I(\psi)$ is finite, and thus $\psi$ is absolutely continuous. Because $L^o(x, \eta, \beta)$ is continuous and finite at each $(x, \eta, \beta) \in \RR^{3}$, a mollification argument allows us to assume that $\dot{\psi}$ is piecewise continuous, as in~\cite[Section 6.5]{DE97}.
Given this~$\psi$ define
$$
  \overline{u}_1(t, x,  y) := -\rho \frac{g(x,y)u_{y}(x,y) \dot{\psi}_t  }{\gamma \overline{Q}(x)}
  \qquad\text{and}\qquad
  \overline{u}_2(t, x,  y) := -\rrho \frac{g(x,y)u_{y}(x,y)\dot{\psi}_t }{\gamma \overline{Q}(x)}.
$$
 Define a control in feedback form by
\begin{equation*}
  \overline{u}^\eps(t) :=  \left( \overline{u}_1(t, \overline{X}_t,  Y_t^\eps),  \overline{u}_2(t, \overline{X}_t,  Y_t^\eps)\right).
\end{equation*}
Then~$R^{\eps,\overline{u}^\eps}$ converges to~$\overline{R}$ in distribution, where
\begin{align*}
  \overline{R}_t & = \int_0^t \left[ \int_\RR -\left[  \gamma^{-1}\rho g(\overline{X}_s, y)u_{y}(\overline{X}_s, y) \overline{u}_1(s, \overline{X}_s,y)  + \gamma^{-1}\rrho g(\overline{X}_s, y)u_{y}(\overline{X}_s, y)\overline{u}_2(s, \overline{X}_s,y) \right] \mu_{ \overline{X}_s}(\D y) \right] \D s \notag \\
  &= \int_0^t \left[ \int_\RR \left[ \gamma^{-2} |g(\overline{X}_s, y)u_{y}(\overline{X}_s, y)|^{2} \right] \mu_{\overline{X}_s}(\D y) \right] \overline{Q}^{-1}(\overline{X}_s)\dot{\psi}_s  \D s \notag \\
  &=  \int_0^t \overline{Q}(\overline{X}_s) \overline{Q}^{-1}(\overline{X}_s)\dot{\psi}_s \D s
  = \int_0^t \dot{\psi}_s \D s = \psi_t, \notag
\end{align*}
with probability one.
If $\overline{Q}(x)$ becomes zero at a countable number of points, we define $\overline{Q}^{-1}(x):=1/\overline{Q}(x)$ if $\overline{Q}(x)\neq 0$ and $0$ otherwise (recall Assumption \ref{C:Positiveness}).
At the same time, the cost satisfies
\begin{equation*}
\lim_{\eps \downarrow 0}\mathbb{E} \left( \frac{1}{2} \int_0^T \lvert \overline{u}_s^\eps \rvert^2 \D s - \frac{1}{2} \int_0^T \int_\RR \lvert \overline{u}(s, \overline{X}_s, y) \rvert^2 \mu_{\overline{X}_s}(\D y) \D s \right)^2 = 0.
\end{equation*}
In addition, the relation~\eqref{Eq:EquivalenceInRep} implies that
\begin{equation*}
\EE\left(\frac{1}{2} \int_0^T \int_\mathcal{Y} \lvert \overline{u}(s, \overline{X}_s,  y) \rvert^2 \mu_{ \overline{X}_s}(\D y) \D s\right)
= \EE\left(I(\overline{R})\right) = I(\psi).
\end{equation*}

Then, we can finally write
\begin{align*}
 \limsup_{\eps \downarrow 0}  - \frac{\log \EE\left[ \exp \left\{ -h^2(\eps) G(R^\eps) \right\} \right]}{h^2(\eps)}
   & = \limsup_{\eps \downarrow 0} \inf_{u \in \mathcal{A}} \mathbb{E} \left[ \frac{1}{2} \int_0^T \lvert u(t) \rvert^2 \D t + G(R^{\eps, u}) \right] \\
  &\le \limsup_{\eps \downarrow 0}  \mathbb{E} \left[ \frac{1}{2} \int_0^T \lvert \overline{u}^\eps(t) \rvert^2 \D t + G(R^{\eps, \overline{u}^\eps}) \right] \\
  & = \mathbb{E} \left[ \frac{1}{2} \int_0^T \int_\RR\lvert \overline{u}(s, \overline{X}_s,  y) \rvert^2 \mu_{\overline{X}_s}(\D y) \D s + G(\overline{R}) \right] \\
  & = [I(\psi) + G(\psi)]
 \le \inf_{\phi \in \mathcal{C}([0,1]; \RR)} [I(\phi) + G(\phi)] + \zeta.
\end{align*}
Since $\zeta >0$ is arbitrary, the upper bound is proved. The proof also shows that we have the explicit representation given by Theorem~\ref{T:IntegratedFunctionsl}.

\appendix

\section{Regularity results on Poisson equations}\label{S:Regularity}

The following theorem collects results from~\cite{PV01} and~\cite{PV03} that are used in this paper.
\begin{theorem}\label{T:regularity}
Under Assumption~\ref{C:ergodic1}, set $F(x, y)=H(x,y)-\overline{H}(x)$.
If there exist $K, q_{F}>0$ such that
\begin{equation*}
   \lvert F(x, y) \rvert + \lVert \partial_x F(x, y) \rVert + \lVert \partial_{xx} F(x, y) \rVert \le K (1 + \lvert y \rvert^{q_{F}}),
\end{equation*}
then there is a unique solution from the class of functions which grow at most polynomially in $\lvert y \rvert$  to
\begin{equation*}
    \Ll_{Y} u(x, y) = - F(x, y),
    \qquad \text{with}\qquad
    \int_\mathcal{Y} u(x, y) \mu(\D y) = 0.
\end{equation*}
Moreover, the solution satisfies $u(\cdot, y) \in \mathcal{C}^2$ for every $y \in \RR$, $\partial_{xx} u \in \mathcal{C}$, and there exists $K'>0$ such that
\begin{align*}
\lvert u(x, y) \rvert +\lVert \partial_y u(x, y) \rVert +\lVert \partial_x u(x, y) \rVert+\lVert \partial_{xx} u(x, y) \rVert&\le K' (1 + \lvert y \rvert )^{q_{F}},
\end{align*}
\end{theorem}

While Theorem~\ref{T:regularity} addresses regularity and growth properties of solutions to general Poisson equations, Lemma~\ref{T:regularityCIR} specialises the discussion on Poisson equations whose operators correspond
to the so-called CEV model---widely used in mathematical finance---and where the right-hand side
$H(x,y)\equiv H(y)$ is only a function of~$y$.
In this case, we have more precise information on the growth of the solution and its derivatives.
\begin{lemma}\label{T:regularityCIR}
Consider the Poisson equation~\eqref{Eq:uPoisson} with the operator given by~\eqref{Eq:CIR}
and $H(x,\cdot) = H(\cdot)\in  \mathcal{C}^{2, \alpha}$
such that for some $K>0$ and $q_{H}\geq 1$,
$\lvert H(y) \rvert  \le K (1 + \lvert y \rvert^{q_{H}})$ holds,
then there is a unique solution from the class of functions which grow at most polynomially in $\lvert y \rvert$  to
\begin{equation*}
    \Ll_{Y} u(y) = H(y)-\overline{H}, \qquad \text{with}\qquad \int_\mathcal{Y} u( y) \mu(\D y) = 0.
\end{equation*}
Moreover, the solution satisfies $u\in \mathcal{C}^2$, and there exists a positive constant $K'$ such that
\begin{align*}
\lvert u(y) \rvert &\le K' (1 + \lvert y \rvert ^{q_{H}}),
\qquad\text{and} \qquad
\lvert \nabla_y u(y) \rvert\le K' (1 + \lvert y \rvert ^{q_{H}-1}).
\end{align*}
If $q_{H}=0$, then
$\lvert u(y) \rvert \le K' (1 + \log(1+\lvert y \rvert))$.
\end{lemma}

\begin{proof}
If the generator  $\Ll_{Y}$ corresponds to the CIR model, i.e., when $q_{g}=1/2$,
we refer the reader to~\cite[Lemma 3.2]{FouqueEtall2011}.
We prove it for the more general case $q_{g}\in(1/2,1)$ and $q_{H}\geq 1$
using a more direct argument.
The invariant measure has density given by the speed measure of the diffusion:
\[
m(y)=\frac{1}{\xi^{2}y^{2q_{g}}}\exp\left\{\int_{1}^{y}\frac{2\kappa(\theta-z)}{\xi^{2}z^{2q_{g}}}\D z\right\},
\]
and a quick computation shows that the derivative to the solution of the Poisson PDE satisfies
\[
u^{\prime}(y)=\frac{1}{\xi^{2}y^{2q_{g}}m(y)}\int_{0}^{y}\left(H(z)-\overline{H}\right)m(z)\D z
 = -\frac{1}{\xi^{2}y^{2q_{g}}m(y)}\int_{y}^{\infty}\left(H(z)-\overline{H}\right)m(z)\D z,
\]
where we used the centering condition to obtain the last equality. Now, without loss of generality we assume that $\xi=\theta=\kappa=1$.
For some positive constant~$C$ that may change from line to line, we can write
\begin{align*}
|u^{\prime}(y)|
 \leq \frac{1}{y^{2q_{g}}m(y)}\int_{y}^{\infty}z^{q_{H}}m(z)\D z
&  \leq C\E^{\frac{y^{2-2q_{g}}}{1-q_{g}}}\int_{y}^{\infty} z^{q_{H}-2q_{g}} \E^{-\frac{z^{2-2q_{g}}}{1-q_{g}}}\D z\\
&  \leq C\E^{R(y)}y^{1+q_{H}-2q_{g}}
\int_{1}^{\infty} g(x) \E^{-y^{2(1-q_{g})}R(x)}\D x,
\end{align*}
with $g(x) := x^{q_{H}-2q_{g}}$ and $R(x):=\frac{x^{2(1-q_{g})}}{1-q_{g}}$.
Since $1-q_g>0$, the function~$R(\cdot)$ is smooth and strictly increasing on $[1,\infty]$.
Furthermore, since $q_H-2q_g \in (0,1)$, the function~$g(\cdot)$ is smooth on $[1,\infty]$.
As~$y$ becomes large, this integral is exactly of Laplace type, albeit without a saddlepoint for~$R(\cdot)$,
its minimum being attained at the left boundary of the domain.
Using~\cite[Section 3.3]{Miller}, we can write, for large~$y$,
$$
\int_{1}^{\infty} g(u) \E^{-y^{2(1-q_{g})}R(u)}\D u
 \sim \frac{g(1)\E^{-y^{2(1-q_g)R(1)}}}{y^{2(1-q_g)}R'(1)}
  = \frac{1}{2}y^{2(q_g-1)}\E^{-y^{2(1-q_g)R(1)}},
$$
and therefore, as~$y$ tends to infinity,
$|u^{\prime}(y)| = \mathcal{O}\left(y^{q_{H}-1}\right)$, and the lemma follows.
\end{proof}
\begin{remark}
It is clear from the proof of Lemma \ref{T:regularityCIR} that its statement is also true in the case of $q_{g}=0$ with $f$ (the drift component in the operator $ \Ll_{Y}$) growing at most linearly in $y$.
\end{remark}
\section{Preliminary results on the involved controlled processes}\label{S:PrelResControlledProcesses}
In this section we recall some results from~\cite{MS} and we prove some additional corollaries related to certain bounds involving the controlled processes~\eqref{Eq:ControlledProcesses} appearing via the weak convergence control representation.

\begin{lemma}[Lemma B.1 of~\cite{MS}] \label{L:uBound}
Under Assumptions~\ref{C:ergodic1} and~\ref{C:StrongSolution1},
let $(X_t^{\eps, u^\eps}, Y_t^{\eps, u^\eps})$ be the strong solution to~\eqref{Eq:ControlledProcesses}. Then the infimum of the representation in~\eqref{Eq:ControlRepresentation} can be taken over all controls such that
\begin{equation*}
    \int_0^T  \lvert u^\eps(s)  \rvert^2 \D s < N, \text{ almost surely},
\end{equation*}
where the constant $N$ does not depend on $\eps$ or $\delta$.
\end{lemma}

\begin{lemma}[Lemma B.2 of~\cite{MS}] \label{L:YIntegralgrowth}
Under Assumptions~\ref{C:ergodic1} and~\ref{C:StrongSolution1},
for $N\in\mathbb{N}$, let $u^{\eps} \in \mathcal{A}$  such that
\begin{equation*}
    \sup_{\eps > 0} \int_0^T \lvert u^{\eps}(s) \rvert^2 \D s < N
\end{equation*}
holds almost surely. Then there exist $\eps_{0}>0$ small enough such that
 \begin{equation*}
    \sup_{\eps\in(0,\eps_{0})}\EE\left(\int_{0}^{T} | Y_s^{\eps, u^{\eps}}|^{2} \D s \right)\le K(N,T),
\end{equation*}
for some finite constant $K(N,T)$ that may depend on $(N,T)$, but not on $\eps,\delta$.
\end{lemma}

Lemmas~\ref{L:Mgrowth}-\ref{L:Ygrowth} follow from Lemma~\ref{L:YIntegralgrowth},
but since their proof was not included in~\cite{MS}, we include them here.
\begin{lemma} \label{L:Mgrowth}
Let Assumptions~\ref{C:ergodic1} and~\ref{C:StrongSolution1} be satisfied. For  $N\in\mathbb{N}$, let $u^{\eps} \in \mathcal{A}$  such that
\begin{equation*}
    \sup_{\eps > 0} \int_0^T \lvert u^{\eps}(s) \rvert^2 \D s < N,
\end{equation*}
almost surely. Define
\[
M^{\eps}_{t} := \frac{\sqrt{\eps}}{\delta}\int_{0}^{t}\E^{-\frac{\eps}{\delta^{2}}\kappa(t-s)}g(X_s^{\eps, u^{\eps}},Y_s^{\eps, u^{\eps}})\D Z_{s}.
\]
Then, for $0\leq q_{g}<1$, and for every $\nu\in(0,1-q_{g})$, there is a constant $C=C(T,\nu)<\infty$ that is independent of $\eps,\delta$ such that
 \begin{equation*}
\EE\left(\sup_{t\in[0,T]}  | M^{\eps}_{T}|^{2}\right)
\leq  C \left(\frac{\sqrt{\eps}}{\delta}\right)^{2(1-\nu)}.
\end{equation*}
\end{lemma}
\begin{proof}[Proof of Lemma~\ref{L:Mgrowth}]
For notational simplicity, let us set $\eta=\sqrt{\eps}/\delta$. By the factorization argument (see for example page 229 of \cite{DaPrato}), we have that for any $\alpha\in(0,1/2)$, there is a constant $C_{\alpha}$ such that
\[
M^{\eps}_{t}= \eta  C_{\alpha}\int_{0}^{t}(t-s)^{\alpha-1}\E^{-\eta^{2}\kappa(t-s)} A^{\eps}_{\alpha}(s)ds,
\]
where
\[
 A^{\eps}_{\alpha}(s)=\int_{0}^{s}(s-r)^{-\alpha}\E^{-\eta^{2}\kappa(s-r)}g(X_r^{\eps, u^{\eps}},Y_r^{\eps, u^{\eps}})\D Z_{r}.
\]

Then for any $p>1/\alpha>2$ we get
\begin{align}
\sup_{t\in[0,T]}  | M^{\eps}_{T}|^{p}&\leq \eta^{p}C^{p}_{\alpha}\left(\int_{0}^{T}s^{\frac{(\alpha-1)p}{p-1}}ds\right)^{p-1}\int_{0}^{T}| A^{\eps}_{\alpha}(s)|^{p}ds\nonumber\\
&=\eta^{p}C^{p}_{p,\alpha}T^{\alpha p-1}\int_{0}^{T}| A^{\eps}_{\alpha}(s)|^{p}ds.\nonumber
\end{align}

By the Burkholder-Davis-Gundy inequality, we get
\begin{align}
\EE \int_{0}^{T}| A^{\eps}_{\alpha}(s)|^{p}ds&\leq C_{p}\int_{0}^{T}\EE\left(\int_{0}^{s}(s-r)^{-2\alpha}\E^{-2\eta^{2}\kappa(s-r)}\left|g(X_r^{\eps, u^{\eps}},Y_r^{\eps, u^{\eps}})\right|^{2}dr\right)^{p/2}ds.\nonumber
\end{align}

Using the inequality that for any $\nu>0$,
\[
\E^{-\lambda t}\leq \left(\frac{\nu}{\E}\right)^{\nu} t^{-\nu}\lambda^{-\nu}, 
\quad\text{for }\lambda,t>0,
\]
we obtain
\[
\E^{-2\eta^{2}\kappa(s-r)}\leq \left(\frac{\nu}{\E}\right)^{\nu} (s-r)^{-\nu}(2\eta^{2}\kappa)^{-\nu}=K(\nu,\kappa) (s-r)^{-\nu}\eta^{-2\nu}.
\]

Therefore, we get
\begin{align}
\EE \int_{0}^{T}| A^{\eps}_{\alpha}(s)|^{p}ds&\leq C_{p}K(\nu,\kappa)^{p/2}\int_{0}^{T}\EE\left(\int_{0}^{s}(s-r)^{-2\alpha-\nu}\eta^{-2\nu}\left|g(X_r^{\eps, u^{\eps}},Y_r^{\eps, u^{\eps}})\right|^{2}dr\right)^{p/2}ds.\nonumber
\end{align}

Next choosing $2\alpha+\nu<1$ and $p>1/\alpha$ we get by Young's inequality
\begin{align}
\EE \int_{0}^{T}| A^{\eps}_{\alpha}(s)|^{p}ds&\leq C_{p}K(\nu,\kappa)^{p/2} \eta^{-\nu p}\left(\int_{0}^{T}s^{-2\alpha-\nu}ds\right)^{p/2}\EE
\int_{0}^{T}\left|g(X_s^{\eps, u^{\eps}},Y_s^{\eps, u^{\eps}})\right|^{p}ds\nonumber\\
&\leq C_{p}K(\nu,\kappa)^{p/2} \eta^{-\nu p}\left(\frac{1}{1-2\alpha-\nu}T^{1-2\alpha-\nu}\right)^{p/2}\EE
\int_{0}^{T}\left(1+\left|Y_s^{\eps, u^{\eps}}\right|^{pq_{g}}\right)ds\nonumber
\end{align}

Then, by Lemma \ref{L:YIntegralgrowth}, if we choose $pq_{g}\leq 2$ (which is possible since $p>1/\alpha>2$ and $q_{g}<1$) we obtain for some constant $C=C(p,T,\alpha,\nu,\kappa)<\infty$
\begin{align}
\EE \int_{0}^{T}| A^{\eps}_{\alpha}(s)|^{p}ds&\leq C \eta^{-\nu p}. \nonumber
\end{align}

Putting the latter bounds together, we then obtain for some constant~$C$ independent of $\eta=\sqrt{\eps}/\delta$:
$$
\EE\left(\sup_{t\in[0,T]}  | M^{\eps}_{T}|^{p}\right)\leq C \eta^{p(1-\nu)}.
$$

Gathering together all the restrictions on $\alpha,\nu,p, q_{g}$, we have that we need $p\in(2,2/q_{g})$, which due to the relation $p>1/\alpha$ means that $\alpha\in(q_{g}/2,1/2)\subset (0,1/2)$. Then the restriction $2\alpha+\nu<1$ gives that $\nu<1-q_{g}$.

Then, for $p>1/\alpha>2$ we get for some unimportant constant $C<\infty$ that may change from line to line
$$
\EE\left(\sup_{t\in[0,T]}  | M^{\eps}_{T}|^{2}\right)
\leq \left(\EE\left(\sup_{t\in[0,T]}  | M^{\eps}_{T}|^{2}\right)^{p/2}\right)^{2/p}\leq \left(\EE\left(\sup_{t\in[0,T]}  | M^{\eps}_{T}|^{p}\right)\right)^{2/p}
\leq C\eta^{2(1-\nu)}=C\left(\frac{\sqrt{\eps}}{\delta}\right)^{2(1-\nu)},
$$
competing the proof of the lemma.
\end{proof}

\begin{lemma} \label{L:Ygrowth}
Let Assumptions~\ref{C:ergodic1} and~\ref{C:StrongSolution1} be satisfied. For  $N\in\mathbb{N}$, let $u^{\eps} \in \mathcal{A}$  such that
\begin{equation*}
    \sup_{\eps > 0} \int_0^T \lvert u^{\eps}(s) \rvert^2 \D s < N,
\end{equation*}
almost surely.
Then, for $0\leq q_{g}<1$, for any $\nu\in(0,1-q_{g})$ and for $\eps>0$ small enough we have
 \begin{equation*}
\EE\left(\sup_{t\in[0,T]}  | Y_t^{\eps, u^{\eps}}|^{2}\right)
\leq  K(N,T,\nu) \left(1+\left(\frac{\sqrt{\eps}}{\delta}\right)^{2(1-\nu)}+h(\eps)^{\frac{2}{1-q_{g}}}\right),
\end{equation*}
for some finite constant $K(N,T)$ that does not depend on $\eps$ and $\delta$.
\end{lemma}
\begin{proof}[Proof of Lemma~\ref{L:Ygrowth}]
By Assumption~\ref{C:ergodic1}, we can write that
$f(x,y)=-\kappa y+\tau(x,y)$,
and we can assume that~$\tau$ is globally Lipschitz in $y$ with Lipschitz constant $L_{\tau}<\kappa$,
where~$\kappa$ is from Assumption~\ref{C:ergodic1}.
For presentation purposes and without loss of generality,
we only prove the case $\rho=0$, i.e., when the Brownian motions are independent.
It is easy to see that
\begin{align*}
Y_t^{\eps, u^{\eps}}&=\E^{-\frac{\eps \kappa}{\delta^{2}}t}y_{0}+\frac{\eps}{\delta^{2}}\int_{0}^{t}\E^{-\frac{\eps}{\delta^{2}}\kappa(t-s)}\tau(X_s^{\eps, u^{\eps}},Y_s^{\eps, u^{\eps}})\D s+ Q^{\eps}_{t}+M^{\eps}_{t},
\end{align*}
where
$$
Q^{\eps}_{t} := \frac{\sqrt{\eps}h(\eps)}{\delta}\int_{0}^{t} \E^{-\frac{\eps}{\delta^{2}}\kappa(t-s)} g(X_s^{\eps, u^{\eps}},Y_s^{\eps, u^{\eps}}) u^{\eps}_{s}\D s
\qquad\text{and}\qquad
M^{\eps}_{t} := \frac{\sqrt{\eps}}{\delta}\int_{0}^{t}\E^{-\frac{\eps}{\delta^{2}}\kappa(t-s)}g(X_s^{\eps, u^{\eps}},Y_s^{\eps, u^{\eps}})\D Z_{s}.
$$
Let us also set $\Lambda^{\eps}_{t} := Y_t^{\eps, u^{\eps}}-Q^{\eps}_{t}-M^{\eps}_{t}$, so that
$$
d\Lambda^{\eps}_{t}=-\frac{\eps}{\delta^{2}}\kappa \Lambda^{\eps}_{t}dt+\frac{\eps}{\delta^{2}} \tau(X_t^{\eps, u^{\eps}},Y_t^{\eps, u^{\eps}})\D t.
$$
Hence, we have
\begin{align}
\frac{1}{2}\frac{\D}{\D t}|\Lambda^{\eps}_{t}|^{2}&= \D\Lambda^{\eps}_{t} \cdot \Lambda^{\eps}_{t}
 = -\frac{\eps}{\delta^{2}}\kappa |\Lambda^{\eps}_{t}|^{2}+\frac{\eps}{\delta^{2}} \tau(X_t^{\eps, u^{\eps}},Y_t^{\eps, u^{\eps}})\Lambda^{\eps}_{t}\nonumber\\
&\leq -\frac{\eps}{\delta^{2}}\kappa |\Lambda^{\eps}_{t}|^{2}+\frac{\eps}{\delta^{2}} \left(\tau(X_t^{\eps, u^{\eps}},\Lambda^{\eps}_{t}+Q^{\eps}_{t}+M^{\eps}_{t})-\tau(X_t^{\eps, u^{\eps}},Q^{\eps}_{t}+M^{\eps}_{t})\right)\Lambda^{\eps}_{t}+\frac{\eps}{\delta^{2}} \tau(X_t^{\eps, u^{\eps}},Q^{\eps}_{t}+M^{\eps}_{t})\Lambda^{\eps}_{t}\nonumber\\
&\leq -\frac{\eps}{2\delta^{2}}(\kappa-L_{\tau}) |\Lambda^{\eps}_{t}|^{2}+K \frac{\eps}{\delta^{2}} \left(1+ |Q^{\eps}_{t}|^{2}+|M^{\eps}_{t}|^2\right),\nonumber
\end{align}
for some unimportant positive constant~$K$.
In the last line we used the Lipschitz and growth properties of~$\tau$ together with Young's inequality.
Then we obtain that for some finite constant $K>0$,
\begin{align*}
|\Lambda^{\eps}_{t}|^{2}&\leq \E^{-\frac{\eps}{\delta^{2}}(\kappa-L_{\tau})t}|y_{0}|^{2}+K \frac{\eps}{\delta^{2}}\int_{0}^{t}\E^{-\frac{\eps}{\delta^{2}}(\kappa-L_{\tau})(t-s)}\left(1+ |Q^{\eps}_{s}|^{2}+|M^{\eps}_{s}|^2\right)\D s.
\end{align*}

We now bound each term separately.  We have for some constant $K$ that may change from line to line
\begin{align*}
\EE\left(\sup_{t\in[0,T]} |Q^{\eps}_{t}|^{2}\right)
 & \leq \frac{\eps h^{2}(\eps)}{\delta^{2}}\EE\sup_{t\in[0,T]}\left|\int_{0}^{t} \E^{-\frac{\eps}{\delta^{2}}\kappa(t-s)} g(X_s^{\eps, u^{\eps}},Y_s^{\eps, u^{\eps}}) u^{\eps}_{s}\D s\right|^{2}\\
&\leq \frac{\eps h^{2}(\eps)}{\delta^{2}}\EE\sup_{t\in[0,T]}\int_{0}^{t} \E^{-2\frac{\eps}{\delta^{2}}\kappa(t-s)} \left|g(X_s^{\eps, u^{\eps}},Y_s^{\eps, u^{\eps}})\right|^{2} \D s \int_{0}^{T}|u^{\eps}_{s}|^{2}\D s\\
&\leq \frac{\eps h^{2}(\eps)}{\delta^{2}} N K\EE \sup_{t\in[0,T]}\int_{0}^{t} \E^{-2\frac{\eps}{\delta^{2}}\kappa(t-s)} \left(1+\left|Y_s^{\eps, u^{\eps}}\right|^{2q_{g}}\right) \D s\\
&\leq \frac{\eps h^{2}(\eps)}{\delta^{2}} N K \sup_{t\in[0,T]}\int_{0}^{t} \E^{-2\frac{\eps}{\delta^{2}}\kappa(t-s)} \D s \left(1+\mathbb{E}\sup_{t\in[0,T]}\left|Y_t^{\eps, u^{\eps}}\right|^{2q_{g}}\right)\\
&\leq K(N,T,\kappa) h^{2}(\eps)\left(1+ \mathbb{E}\sup_{t\in[0,T]}\left|Y_t^{\eps, u^{\eps}}\right|^{2q_{g}}\right).
\end{align*}

Using now Lemma \ref{L:Mgrowth}, we have that for any $\nu\in(0,1-q_{g})$,
$$
\EE\left(\sup_{t\in[0,T]} |M^{\eps}_{t}|^{2}\right)
 \leq C \left(\frac{\sqrt{\eps}}{\delta}\right)^{2(1-\nu)}.
$$
Consequently, we obtain that for some appropriate constant $K<\infty$,
\begin{align*}
\EE\left(\sup_{t\in[0,T]} |\Lambda^{\eps}_{t}|^{2}\right)
 & \leq K\left\{1 + h^{2}(\eps)+
 \left(\frac{\sqrt{\eps}}{\delta}\right)^{2(1-\nu)}
 + h^{2}(\eps) \EE\left(\sup_{t\in[0,T]}\left|Y_t^{\eps, u^{\eps}}\right|^{2q_{g}} \right)\right\}
\end{align*}

Hence, recalling the original definition of $\Lambda^{\eps}_{t}=Y_t^{\eps, u^{\eps}}-Q^{\eps}_{t}-M^{\eps}_{t}$, we have for $\eps>0$ small enough such that $h(\eps)>1$ and for some constant $K<\infty$ that changes from line to line,
\begin{align*}
\EE\left(\sup_{t\in[0,T]}  | Y_t^{\eps, u^{\eps}}|^{2}\right)
 & \leq K\left(1 + h^{2}(\eps)+\left(\frac{\sqrt{\eps}}{\delta}\right)^{2(1-\nu)}
 + h^{2}(\eps)\EE\left(\sup_{t\in[0,T]}\left|Y_t^{\eps, u^{\eps}}\right|^{2q_{g}}\right)\right)\\
&\leq K\left(1+ h^{2}(\eps)+\left(\frac{\sqrt{\eps}}{\delta}\right)^{2(1-\nu)}+  h^{2}(\eps)\EE\left(\sup_{t\in[0,T]}\left|Y_t^{\eps, u^{\eps}}\right|^{2q_{g}}\right) \right)\\
&\leq K\left(1+ h^{2}(\eps)+\left(\frac{\sqrt{\eps}}{\delta}\right)^{2(1-\nu)}+  h(\eps)^{\frac{2}{1-q_{g}}}\right)
 + \frac{1}{2}\EE\left(\sup_{t\in[0,T]}\left|Y_t^{\eps, u^{\eps}}\right|^{2}\right),
\end{align*}
where in the last step we used Young's generalised inequality.
Hence, by rearranging terms we obtain
$$\EE\left(\sup_{t\in[0,T]}  | Y_t^{\eps, u^{\eps}}|^{2}\right)
\leq K\left(1+ \left(\frac{\sqrt{\eps}}{\delta}\right)^{2(1-\nu)}+h(\eps)^{\frac{2}{1-q_{g}}}\right),$$
completing the proof of the lemma.
\end{proof}

\begin{lemma}[Lemma B.3 in~\cite{MS}] \label{L:productBound}
Let Assumption~\ref{C:ergodic1} hold,
$N\in\mathbb{N}$, and $u^\eps = (u_1^\eps, u_2^\eps) \in \mathcal{A}$ such that
\begin{equation*}
    \sup_{\eps > 0} \int_0^T  \lvert u^\eps(s) \rvert^2  \D s < N
\end{equation*}
holds almost surely.
Let $A(x, y)$ and $B(x,y)$ be given functions and $K$, $\theta\in(0,1)$ such that
\begin{equation*}
    |A(x, y)| \le K (1 + \lvert y \rvert^\theta)
    \qquad\text{and} \qquad
    |B(x, y)| \le K (1 + \lvert y \rvert^{2\theta}).
\end{equation*}
Then for $\alpha \in \{1, 2\}$,
\begin{enumerate}[(i)]
\item for any $p\in(1,1/\theta]$, there exists $C>0$ such that
$$
\EE\left(\sup_{t\in[0,T]}\left\lvert \int_{0}^{T} A(X_s^{\eps, u^\eps}, Y_s^{\eps, u^\eps}) u_\alpha^\eps(s) \D s \right\rvert^{2p}
+ \sup_{t\in[0,T]}\left\lvert \int_{0}^{T} B(X_s^{\eps, u^\eps}, Y_s^{\eps, u^\eps})  \D s \right\rvert^{p}\right)\le C;
$$
\item for any $p\in(1,1/\theta]$, there exists $C>0$ such that for fixed $e > 0$ and for all $0 \le t_1 < t_1+e \le T$,
$$
\EE\left(\sup_{\substack{0 \le t_1 < t_2 \le T \\ \lvert t_2 - t_1 \rvert < e}}\left\lvert \int_{t_1}^{t_{2}} A(X_s^{\eps, u^\eps}, Y_s^{\eps, u^\eps}) u_\alpha^\eps(s) \D s \right\rvert^{2p}
+
\sup_{\substack{0 \le t_1 < t_2 \le T \\ \lvert t_2 - t_1 \rvert < e}}\left\lvert \int_{t_1}^{t_{2}} B(X_s^{\eps, u^\eps}, Y_s^{\eps, u^\eps})  \D s \right\rvert^{p}\right)
\le C \lvert e \rvert^{r/\theta-1};
$$
\item for all $\zeta > 0$,
\begin{equation*}
    \lim_{e\downarrow 0} \limsup_{\eps \downarrow 0} \mathbb{P} \left[ \sup_{\substack{0 \le t_1 < t_2 \le T \\ \lvert t_2 - t_1 \rvert < e}} \left\lvert \int_{t_1}^{t_2} A(X_s^{\eps, u^\eps}, Y_s^{\eps, u^\eps}) u_\alpha^\eps(s) \D s \right\rvert > \zeta \right] = 0
\end{equation*}
and
\begin{equation*}
    \lim_{e \downarrow 0} \limsup_{\eps \downarrow 0} \mathbb{P} \left[ \sup_{\substack{0 \le t_1 < t_2 \le T \\ \lvert t_2 - t_1 \rvert < e}} \left\lvert \int_{t_1}^{t_2} B(X_s^{\eps, u^\eps}, Y_s^{\eps, u^\eps})  \D s \right\rvert > \zeta \right] = 0.
\end{equation*}
\end{enumerate}
\end{lemma}



\begin{thebibliography}{99}

\bibitem{BaierFreidlin}D. Baier and M. I. Freidlin.
Theorems on large deviations and stability for random perturbations.
\textit{Dokl. Akad. Nauk SSSR} {\tt 235}: 253-256, 1977.

\bibitem{Billingsley1968}P. Billingsley.
Convergence of Probability Measures, Second Edition, Wiley, New York, 1968.

\bibitem{BD98}M. Bou\'{e} and P. Dupuis.
A variational representation for certain functionals of Brownian motion.
\textit{Annals of Probability}, {\tt 26}(4): 1641-1659, 1998.

\bibitem{JCY}M. Chesney, M. Jeanblanc and M. Yor.
Mathematical methods for financial markets.
Springer-Verlag London, 2009.

\bibitem{ChiariniFischer2014}A. Chiarini and M. Fischer.
On large deviations for small noise It\^{o} processes.
\textit{Advances App. Prob.}, {\tt 46}: 1126-1147, 2014.

\bibitem{DaPrato}	G. Da Prato and J. Zabczyk.
Stochastic Equations in Infinite Dimensions.   Encyclopedia of Mathematics and its Applications. Cambridge University Press, 2008.

\bibitem{DZ98}A. Dembo and O. Zeitouni.
Large deviations techniques and applications.
Springer-Verlag New-York, 1998.

\bibitem{DFJV1}J.D. Deuschel, P.K. Friz, A. Jacquier and S. Violante.
Marginal density expansions for diffusions and stochastic volatility, Part I: Theoretical foundations.
\textit{Communications on Pure and Applied Mathematics}, {\tt 67}(1): 40-82, 2014.

\bibitem{DFJV2}J.D. Deuschel, P.K. Friz, A. Jacquier and S. Violante.
Marginal density expansions for diffusions and stochastic volatility, Part II: Applications.
\textit{Communications on Pure and Applied Mathematics}, {\tt 67}(2): 321-350, 2014.

\bibitem{DE97}P. Dupuis and R.S. Ellis.
A weak convergence approach to the theory of large deviations.
John Wiley \& Sons, New York, 1997.

\bibitem{DS12}P. Dupuis and K. Spiliopoulos.
Large deviations for multiscale problems via weak convergence methods.
\textit{Stochastic Processes and their Applications}, {\tt 122}: 1947-1987, 2012.

\bibitem{FFK12}J. Feng, J.-P. Fouque and R. Kumar.
Small-time asymptotics for fast mean-reverting stochastic volatility models.
\textit{Annals of Applied Probability}, {\tt 22}(4): 1541-1575, 2012.

\bibitem{FordeJac}M. Forde and A. Jacquier.
The large-maturity smile for the Heston model
\textit{Finance and Stochastics}, {\tt 15}(4): 755-780, 2011.

\bibitem{Freidlin1978}M. Freidlin.
The averaging principle and theorems on large deviations.
\textit{Russian Math. Surveys}, {\tt 33}(5): 117-176, 1978.

\bibitem{FrizPigato}P. Friz, P. Gassiat and P. Pigato.
Precise asymptotics: robust stochastic volatility models.
Preprint \href{https://arxiv.org/abs/1811.00267}{arXiv: 1811.00267}, 2018.

\bibitem{FGGJT}P.~Friz, J.~Gatheral, A.~Gulisashvili, A.~Jacquier and J.~Teichmann.
Large deviations and asymptotic methods in finance.
Springer Proceedings in Mathematics and Statistics, {\tt 110}, 2015.

\bibitem{FGP}P. Friz, S. Gerhold and A. Pinter.
Option pricing in the moderate deviations regime.
Forthcoming in \textit{Mathematical Finance}.

\bibitem{FGY}P. Friz, S. Gerhold and M. Yor.
How to make Dupire's local volatility work with jumps.
\textit{Quant. Fin.}, {\tt 14}(8): 1327-1331,~2014.

\bibitem{FouqueEtall2011}J.-P, Fouque, G. Papanicolaou, R. Sircar and K. Solna.
Multiscale stochastic volatility for equity, interest rate, and credit derivatives.
Cambridge University Press, Cambridge, 2011.

\bibitem{G03}A.~Guillin.
Averaging principle of SDE with small diffusion: moderate deviations.
\textit{Annals of Probability}, {\tt 31}(1): 413-443, 2003.

\bibitem{Guillin}A. Guillin and R. Liptser.
Examples of moderate deviation principle for diffusion processes.
\textit{Discrete and Continuous Dynamical Systems (B)}, {\tt 6}(4): 803-828, 2006.

\bibitem{GuillinLipster2005}A. Guillin and R. Liptser.
MDP for integral functionals of fast and slow processes with averaging.
\textit{Stochastic Processes and their applications}, {\tt 115}: 1187-1207, 2005.

\bibitem{GuliStein}A. Gulisashvili and E. Stein.
Asymptotic behavior of the stock price distribution density and implied volatility in stochastic volatility models.
\textit{Applied Mathematics and Optimization}, {\tt 61}(3): 287-315.

\bibitem{JKRM}A. Jacquier, M. Keller-Ressel and A. Mijatovi\'c.
Large deviations and stochastic volatility with jumps.
\textit{Stochastics}, {\tt 85}(2): 321-345, 2013.

\bibitem{KarShreve}I. Karatzas and S. Shreve.
Brownian Motion and Stochastic Calculus. Springer, 1991.

\bibitem{Lipster}R. Lipster and V. Spokoiny.
Moderate deviations type evaluation for integral functional of diffusion processes.
\textit{Electronic Journal of Probability}, {\tt 4}: 1-25, 1999.

\bibitem{Miller}P.D. Miller.
Applied asymptotic analysis.
Graduate Studies in Mathematics, {\tt 75}, AMS, 2006.

\bibitem{MS}M. Morse and K. Spiliopoulos.
Moderate deviations principle for systems of slow-fast diffusions.
\textit{Asymptotic Analysis}, {\tt 105}: 97-135, 2017.

\bibitem{PV01}E.~Pardoux and A.Yu.~Veretennikov.
On the Poisson equation and diffusion approximation I.
\textit{Annals of Probability}, {\tt 29}(3): 1061-1085, 2001.

\bibitem{PV03}E.~Pardoux and A.Yu.~Veretennikov.
On Poisson equation and diffusion approximation II.
\textit{Annals of Probability}, {\tt 31}(3): 1166-1192, 2003.

\bibitem{Robertson}S. Robertson.
Sample path large deviations and optimal importance sampling for stochastic volatility models.
\textit{Stochastic Processes and their Applications}, {\tt 120}(1), 66-83, 2010.

\bibitem{S13}K. Spiliopoulos.
Large deviations and importance sampling for systems of slow-fast motion.
\textit{Applied Mathematics and Optimization}, {\tt 67}: 123-161, 2013.

\bibitem{Stein}E. Stein and J. Stein.
Stock price distributions with stochastic volatility: an analytic approach.
\textit{Review of Financial Studies}, {\tt 4}(4): 727-752, 1991.
\end{thebibliography}
\end{document}